\documentclass[journal]{IEEEtran}
\IEEEoverridecommandlockouts
\usepackage{amsthm}
\usepackage{amsmath}
\usepackage{amsfonts}
\usepackage{amssymb}

\usepackage{graphicx}
\usepackage{enumerate}
\usepackage{subfigure}
\usepackage{comment}
\usepackage{newtxmath}
\usepackage{bm}

\usepackage{cite}
\usepackage{color}
\usepackage{booktabs, soul}

\usepackage{algorithm}  
\usepackage{algorithmicx} 
\usepackage{algpseudocode} 
\usepackage{subfigure}

\usepackage[hyphens]{url}
\usepackage[bookmarks,colorlinks]{hyperref}  
\usepackage[hyphenbreaks]{breakurl}

\newtheorem{mythm}{Theorem}

\title{Joint Transmit Beamforming for Multiuser MIMO Communications and Radar
}
\author{ 
	Xiang  Liu, Tianyao Huang, Nir Shlezinger, Yimin Liu, Jie Zhou,  and Yonina C. Eldar
	\thanks{
		This work was supported by the National Natural Science Foundation of
		China under Grant 61801258 and 61571260, Futurewei Technologies, and the Air Force Office of Scientific Research under grant No. FA9550-18-1-0208.		
		Part of the work \cite{mywork} were presented at 2019 IEEE International Conference on Signal, Information and Data Processing.	
		X. Liu, T. Huang and Y. Liu are with the Department of Electronic Engineering, Tsinghua University, Beijing, China (e-mail: liuxiang16@mails.tsinghua.edu.cn, \{huangtianyao; yiminliu\}@tsinghua.edu.cn). 		
		J. Zhou is with the Institute of Electronic Engineering, China Academy of Engineering Physics, Mianyang, China (e-mail: zhoujie\_iee@yeah.net).		
		N. Shlezinger  and Y. C. Eldar are with the Faculty of Mathematics and Computer Science, Weizmann Institute of Science, Rehovot, Israel (e-mail: \{nir.shlezinger; yonina\}@weizmann.ac.il). 	
		
	}
}

\begin{document}
	\maketitle
	\begin{abstract}
		Future wireless communication systems are expected to explore spectral bands typically used by radar systems, in order to overcome spectrum congestion of traditional communication bands.  Since in many applications radar and communication share the same platform, spectrum sharing can be facilitated by joint design as dual function radar-communications system.
		%
		In this paper, we propose a joint transmit beamforming model for a dual-function multiple-input-multiple-output (MIMO) radar and multiuser MIMO communication transmitter sharing the spectrum and an antenna array.
		The proposed dual-function system transmits the weighted sum of independent radar waveform and communication symbols, forming multiple beams towards the radar targets and the communication receivers, respectively.
		The design of the weighting coefficients is formulated as an optimization problem whose objective is the performance of the MIMO radar transmit beamforming, while guaranteeing that the signal-to-interference-plus-noise ratio (SINR)  at each communication user is higher than a given threshold.
		Despite the non-convexity of the proposed optimization problem,  it can be relaxed into a convex one, which can be solved in polynomial time, and we prove that the relaxation is tight. Then, we propose a reduced complexity design based on  zero-forcing the inter-user interference and radar interference.
		Unlike previous works, which focused on the transmission of communication symbols to synthesize a radar transmit beam pattern, our method provides more degrees of freedom for MIMO radar and is thus able to obtain improved radar performance, as demonstrated in our simulation study.
		Furthermore, the proposed dual-function scheme approaches the radar performance of the radar-only scheme, i.e., without spectrum sharing, under reasonable communication quality constraints. 
	\end{abstract}

	
	\section{Introduction}
	The increasing  demands on wireless communications networks give rise to a growing need for spectrum sharing between radar and communication systems. 
	Nowadays, military radars utilize numerous spectrum bands below 10 GHz, like S-band (2-4 GHz) and C-band (4-8 GHz), while spectrum congestion is becoming a serious problem which limits the throughput of wireless communications operating in neighbouring bands. To tackle this congestion, it has  been recently proposed to allow wireless communications to share spectrum with radar systems, allowing both functionalities to simultaneously operate over the same wide frequency bands \cite{paul_survey_2017-1, darpa, noauthor_fcc_nodate, zheng2019radar}.
	
	The common strategy to allow individual radar and communication systems to share spectrum with controllable mutual interference is to facilitate co-existence by some level of cooperation \cite{6167362, 6331681, sodagari_projection_2012, mahal_spectral_2017, deng_interference_2013, meager_estimation_2016,nartasilpa_lets_2018,nartasilpa_communications_2018, rihan_optimum_2018, qian_joint_2018, li_optimum_2016-1, li_joint_2017, liu_mimo_2018}. 
	These techniques include opportunistic spectrum access \cite{6167362, 6331681}, transmit interference nulling \cite{sodagari_projection_2012, mahal_spectral_2017}, adaptive receive interference cancellation \cite{deng_interference_2013, meager_estimation_2016,nartasilpa_lets_2018,nartasilpa_communications_2018} and optimization based beamforming design \cite{rihan_optimum_2018, qian_joint_2018, li_optimum_2016-1, li_joint_2017, liu_mimo_2018} to mitigate the mutual interference.
	This approach typically requires the individual   radar and communication systems to either be configured using some centralized entity, or alternatively, to exchange information, such as knowledge of the interference channel and radar waveform parameters, significantly increasing the complexity of realizing such  systems \cite{liu_toward_2018}.
	
	The difficulty associated with coordinating spectrum sharing radar and communication systems is notably reduced when these functionalities operate on the same device. In fact, various emerging technologies, such as automotive vehicles \cite{ma2019joint}, implement both radar sensing and data transmission from the same platform, as illustrated in Fig.~\ref{fig:function}. In such cases, spectrum sharing can be realized by jointly designing a  dual-function radar-communications (DFRC) system \cite{7409935, ma2019joint, sturm_novel_2009, sturm_ofdm_2009, sturm_waveform_2011, zhang_ofdm_2015, liu_application_2017,sahin_novel_2017, li_integrated_2019, sahin_filter_2017, zhang_modified_2017, ISI:000453637100020, ISI:000453637100011, wang_dual-function_2018,hassanien_dual-function_2016, 7575457, hassanien_dual-function_2016-1, Ma2018, Huang2019multi, Huang2019A, mccormick_simultaneous_2017-1, mccormick_simultaneous_2017,liu_toward_2018, liu_mu-mimo_2018}. 
	One clear advantage of DFRC methods over individual co-existing systems is that the functionalities share radio frequency (RF) front-end and aperture, thus reducing the cost and weight of hardware \cite{kumari_ieee_2018}.
	Moreover, radar and communication are naturally combined in a DFRC system, and no additional cost is required for cooperation. 
	Nonetheless, DFRC design has several associated challenges. From a hardware perspective, the requirements of radar and communications may be quite distinct in terms of, e.g., power amplifiers operation mode \cite{sturm_waveform_2011, liu_toward_2018}.
	From the algorithmic side, properly combining radar and communications is a challenging task, and a broad range of strategies for doing so have been proposed in the literature, see, e.g., the detailed survey in \cite{ma2019joint}.
	

	\begin{figure}
		\centering
		\includegraphics[width=0.9\linewidth]{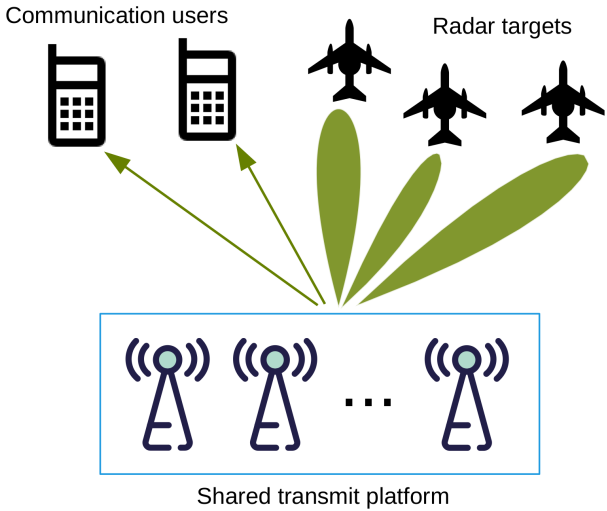}
		\caption{A dual function system in which communication and radar share the transmit platform.} 
		\label{fig:function}
	\end{figure}

	Early works on DFRC systems consider single-antenna devices. 
	One way to implement such spectrum-sharing dual function signaling is by utilizing orthogonal individual signals for radar and communications, as proposed in 
	\cite{7409935} which studied time-division based DFRC systems.
	Alternatively, one can  achieve  both functions simultaneously by employing an appropriate integrated waveform, which can be utilized for both target detection and information transmission.
	For instance, the probing capabilities of orthogonal frequency division multiplexing (OFDM) waveforms, which are widely used for communication signaling, were studied in  \cite{sturm_novel_2009, sturm_ofdm_2009, sturm_waveform_2011, zhang_ofdm_2015, liu_application_2017}.
	The combination of linear frequency modulation (LFM), which is a traditional radar waveform, with continuous phase modulation (CPM) to realize a dual-function signal capable of conveying information was studied \cite{sahin_novel_2017, li_integrated_2019, sahin_filter_2017, zhang_modified_2017}. 
	However, these schemes inherently result in  performance loss for either the radar or the communication \cite{liu_toward_2018}.
	For instance, LFM-CPM usually exhibits higher side-lobes than standard LFM \cite{sahin_novel_2017}.
	Moreover, a common problem emerging in these single-antenna schemes is that radar systems with integrated waveforms usually form a single directional beam, which illuminates the radar target inside the beam.
	Therefore, single-antenna schemes are not able to illuminate multiple targets and communicate with multiple users simultaneously as in Fig.~\ref{fig:function}. That leads to notable degradation in signal-to-noise-ratio (SNR) when the communication receivers are not physically located within the  radar main lobe.

	
	Recognizing this limitation of single-antenna schemes, recent works on DFRC methods employ multiple-input-multiple-output (MIMO) systems, which provide higher degrees of spatial
	freedom, and can simultaneously synthesize multiple beams towards several communication users and radar targets. 
	These  studies can be divided into two categories: information embedding \cite{ISI:000453637100020, ISI:000453637100011, wang_dual-function_2018,hassanien_dual-function_2016, 7575457, hassanien_dual-function_2016-1, Ma2018, Huang2019A} and transmit beamforming \cite{liu_toward_2018, mccormick_simultaneous_2017-1, mccormick_simultaneous_2017, liu_mu-mimo_2018}.
	In information embedding systems, radar is typically considered  as the primary function, and the communication message is encoded into the MIMO radar waveform.
	For example, the works \cite{hassanien_dual-function_2016, 7575457} proposed to embed communication bits by  controlling the amplitude and phase of radar spatial side-lobes.
	The works in \cite{ISI:000453637100020, ISI:000453637100011, wang_dual-function_2018} proposed to convey the message in the form of index modulation, via the  selection of  active radar transmit antennas and the allocation of  radar waveforms across active antennas. 
	Embedding data bits in the parameters of a radar waveform, e.g., phase, antenna index, and frequency, can yield communications in the form of phase modulation \cite{hassanien_dual-function_2016-1}, spatial modulation \cite{Ma2018}, and carrier frequency modulation \cite{Huang2019A}.
	However, such methods carry a very limited number of communication symbols per radar pulse, usually yielding  low information rate of the same order of radar pulse repetition frequency \cite{liu_toward_2018}.
	
	The second approach for implementing MIMO DFRC systems is based on transmit beamforming. Here, the spatial degrees of freedom is exploited to synthesize multiple beams towards several communication users and radar targets. 
	As opposed to information embedding strategies, transmit beamforming  enables each function to use its individual waveform, potentially supporting higher data rates and guaranteed radar performance by utilizing conventional dedicated signals for each functionality  \cite{liu_toward_2018}. 
	In this approach, the main design goal is to properly beamform both the radar and communication signals such that each can operate reliably.
	
	In \cite{mccormick_simultaneous_2017-1, mccormick_simultaneous_2017}, the array probing signal is designed to synthesize radar and communication waveform towards different directions.
	The method in \cite{mccormick_simultaneous_2017-1, mccormick_simultaneous_2017} considers waveform synthesis at the main beam direction but cannot suppress the  azimuth side-lobe of the radar transmit beam pattern.
	The work \cite{liu_toward_2018} extended the method of \cite{mccormick_simultaneous_2017-1, mccormick_simultaneous_2017} to designing the array probing signal to match the radar transmit beam pattern and minimize the interference power at multiple users.
	However, the methods in \cite{mccormick_simultaneous_2017-1, mccormick_simultaneous_2017, liu_toward_2018} only minimize the interference power, and do not consider the signal-to-interference-plus-noise ratio (SINR) at each user, which is the quantity dictating the communications rate. 
	In \cite{liu_mu-mimo_2018}, the authors studied transmit beamforming in DFRC systems with multiple receivers, i.e., multiuser setup, in which the communication waveform is utilized as a radar transmit waveform. 
	In such a dual-function system, the available degrees of freedom (DoF) for the MIMO radar waveform, which affects the resulting radar beam pattern, is equal to the number of communication users. Since the DoF of conventional MIMO radar, i.e., without communication functionality, is at most the number of transmit antennas, the resulting MIMO radar cannot utilize its full DoF when the number of users is smaller than the number of antennas, potentially   leading to significant distortion of  radar beam pattern. 
	Furthermore, the resulting problem is a non-convex optimization problem, which is solved by sub-optimal methods.
	
	In this paper, we design a transmit beamforming based MIMO DFRC system. 
	As was done in \cite{liu_mu-mimo_2018}, we design our transmit beamforming to optimize both the radar transmit beam pattern and the SINR at the communication users. 
	Unlike \cite{liu_mu-mimo_2018}, our proposed joint transmitter utilizes jointly precoded {\em individual communication and radar waveforms}, allowing to extend the MIMO radar waveform DoF to its maximal value, i.e., the number of antennas.  
	In fact, the previously proposed formulation of \cite{liu_mu-mimo_2018} can be  regarded as a special case of the proposed one by nullifying the dedicated radar waveform. 
	Furthermore, while we utilize individual signals for radar and communications, we exploit the fact that both signals are transmitted from the same device, which also accommodates the radar receiver. Consequently, the radar receiver has complete knowledge of the transmitted communication waveform, which is  utilized for target detection in addition to the dedicated radar signal. This approach contributes to the power-efficiency of the DFRC system, further exploiting the inherent advantages of joint design over co-existing separate radar and communication systems.

	We formulate the design of the resulting precoding method as a non-convex optimization problem, which we tackle using two different methods: First, we show that it can be relaxed into  an equivalent semidefinite problem, where the latter can be solved using conventional optimization tools, and prove that the relaxation is tight. To circumvent the computational burden of recovering the optimal precoders from the relaxed formulation, we propose an additional design approach based on zero-forcing the interference. 
	Our numerical results demonstrate that, due to the increased DoF of the MIMO radar waveform, our approach  obtains improved radar transmit beam pattern compared to \cite{liu_mu-mimo_2018}, under the same SINR constraints at the communication users. Furthermore, we demonstrate that under high SINR constraints our reduced complexity zero-forcing technique is capable of achieving comparable performance to that of the optimal beamforming scheme, whose compuataion is substantially more complex, indicating the potential of our approach in designing reliable DFRC beamforming at controllable complexity.

	The rest of this paper is organized as follows.
	Section~\ref{section2} proposes the system  model.
	Section~\ref{section3} introduces the performance metrics of MIMO radar and multiuser MIMO communication, respectively. 
	Section~\ref{section4} establishes the optimization model for joint beamforming, and proposes two methods for designing precoders based on that formulation. 
	Simulation results are presented in Section~\ref{section5}. Finally, Section~\ref{sec:conclusion} provides concluding remarks.

	\emph{Notations:} 
	In this paper,  $(\cdot) ^ {H}$, $(\cdot) ^ c$ and $(\cdot) ^ {T}$ denote  Hermitian transpose, conjugate and transpose, respectively. 
	Vectors are denoted by bold lower class letters and matrices are denoted by bold upper class letters.
	For a matrix $\bm{A}$, the $(i,j)$-th elements of $\bm{A}$ is denoted by $\left[\bm{A}\right]_{i,j}$, and $\left[\bm{A}\right]_{1:j}$ denotes the sub-matrix containing the first $j$ columns of $\bm{A}$. 
	We let $\bm I_n$ and $\bm 0_{m \times n}$ denote $n$-dimensional identity matrix and $m \times n$ zero matrix, respectively.
	We use $\mathbf{E}(\cdot)$ for the stochastic expectation.
	For an integer $n > 0$, the set consisting of all $n$-dimensional complex positive semidefinite matrices is denoted by $\mathcal{S}_n ^ +$.

	\section{Signal  Model} \label{section2}

	Consider an antenna array shared by a colocated monostatic MIMO radar system and a multiuser MIMO communication transmitter as depicted in Fig.~\ref{fig:function}. 
	In our work, both functionalities operate simultaneously by joint beamforming.
	The system diagram of our joint beamforming transmitter is shown in Fig.~\ref{fig:system}, demonstrating that the  transmited signal is a weighted sum of communication symbols and radar waveform.
	We consider an antenna array of $M$ elements, and let the discrete-time transmit signal of this array at time index $n$ be given by 
	\begin{equation} \label{eq-1-1}
	\bm{x}[n] = \bm{W}_r \bm{s} [n] +  \bm{W}_c \bm{c} [n],\ n = 0,\ldots,N-1.
	\end{equation}
	Here, the $M\times 1$ vector $\bm{s} [n] = [s_1[n], \ldots, s_{M}(n)]^T$ includes $M$ individual radar waveforms, 
	and the $M \times M$ matrix $ \bm{W}_r $ is the beamforming matrix (or precoder) for radar waveform.
	Similarly, $\bm{c}[n] =  [c_{1}[n], \ldots, c_{K}[n]]^T$ is a $K\times 1$  vector including $ K $  parallel communication symbol streams to be communicated to $K$ users, respectively, while the $M \times K$ matrix $ \bm{W}_c $ is the communication precoder.
	To achieve  alias-free signal sampling and symbol transmission, the communication symbol duration or the radar code duration, denoted by $t_s$, should satisfy $t_s \geq  1/(2B)$  \cite{989875}, where $B$ is the baseband bandwidth of the transmit platform. 
	The maximal available radar delay resolution is  bounded by the  symbol duration, and is given by $1/(2B)$. The maximal available symbol rate is $2B$.
	We note that the scheme proposed in \cite{liu_mu-mimo_2018}, which beamformed the communications symbols to be utilized for probing, can be regarded as a special case of our system by letting the radar waveform be zero, namely transmitting only communication symbols, as depicted in Fig.~\ref{fig:system3}.

	\begin{figure} 
		\centering
		\subfigure[]{
			\includegraphics[width=\linewidth]{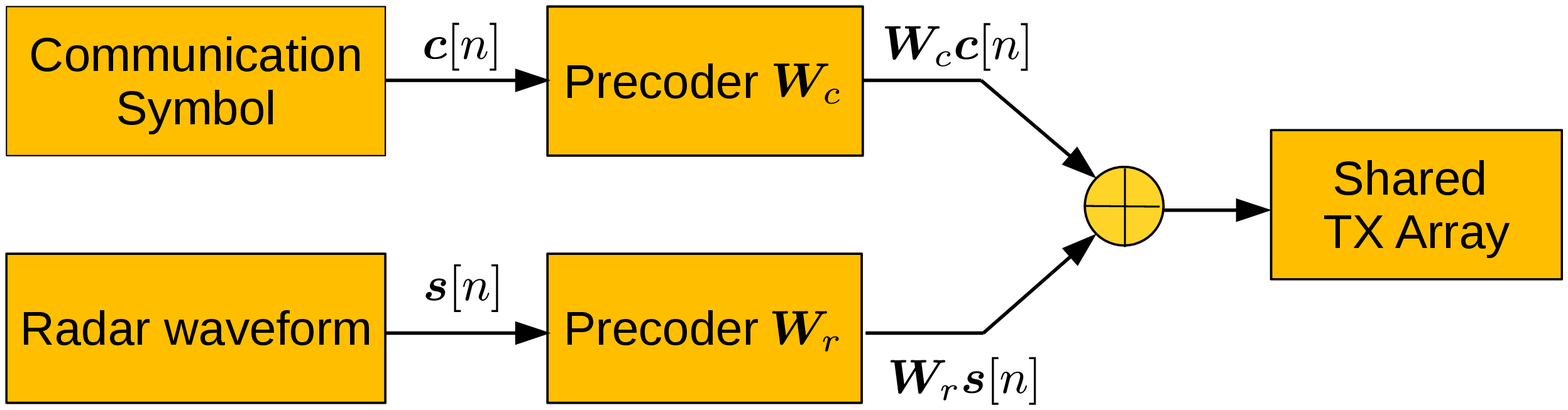}
			\label{fig:system}
		}
		\subfigure[]{
			\includegraphics[width=\linewidth]{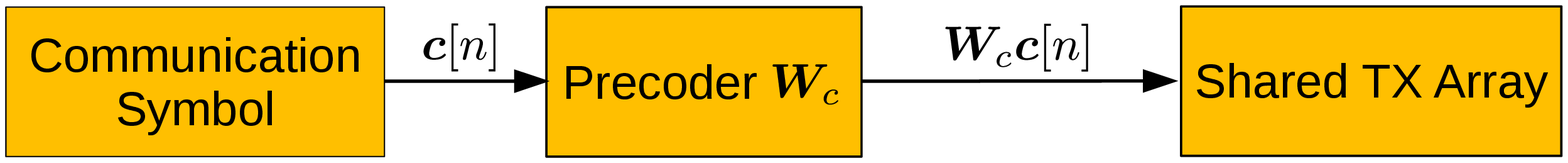}
			\label{fig:system3}
		}
		\caption{(a) The joint transmitter jointly precodes  both communication symbols and radar waveform; (b) The joint transmitter only utilizes precoded  communication symbols.}
	\end{figure}
	
	Our goal is to design the matrices $W_c$ and 
	$W_r$ in \eqref{eq-1-1}. Without loss of generality, we make the following assumptions: 1) Both radar and communication signals are zero-mean, temporally-white and wide-sense stationary stochastic process; 2) The communication symbols are uncorrelated with radar waveform, i.e.,
	\begin{equation}  \label{eq-1-3}
	\mathbf{E} \left( \bm{s}[n] \bm{c} ^ H [n] \right) = \mathbf{E} \left( \bm{s}[n] \right) \mathbf{E} ( \bm{c} ^ H [n] ) = \bm{0}_{M \times K};
	\end{equation}
	3) Communication symbols intended to different users are   uncorrelated, namely,
	\begin{equation} \label{eq-1-4}
	\mathbf{E} \left( \bm{c}[n] \bm{c} ^ H [n]  \right) = \bm{I}_K;
	\end{equation}
	4) The individual radar waveform are generated by pseudo random coding \cite{6071523, 8468331, 6176764, 1455966, golomb2005signal}, and thus are  uncorrelated with each other, resulting in 
	\begin{equation} \label{eq-1-2}
	\mathbf{E} \left( \bm{s}[n] \bm{s} ^ H [n]  \right) = \bm{I}_M.
	\end{equation}
	Here, both signals are normalized to have unit power, and their real power is encapsulated in their corresponding precoders $\bm{W}_r$ and $\bm{W}_c$.
	
	Since we focus on DFRC schemes in which individual uncorrelated waveforms are used for radar and communications, radar interference can induce a notable degradation in the ability of the communication receivers to recover the transmitted symbols. In particular, to achieve radar detection, the transmit power of the dual function platform is usually much higher than that of a typical communication transmitter, because the echos reflected from the targets are attenuated with a two-way propagation loss. However, radar transmit beams are designed to be highly directional, and thus radar interference is dominant only when the radar beam is steered towards the communication receivers.  
	In particular, when the communication receivers lie in the same direction as that of the radar targets, they are expected to observe high radar interference.
	Nevertheless, the transmitted communication waveform is completely known at the  radar receiver  and thus its reflected signal  can also be utilized for target detection, indicating that this challenge can be overcome by forming a high-gain communication beam to  simultaneously cover the targets and communication receiver for jointly radar sensing and data transmission, as in the shared integrated waveform DFRC  design \cite{ sturm_novel_2009, sturm_ofdm_2009, sturm_waveform_2011, zhang_ofdm_2015, liu_application_2017 ,sahin_novel_2017, li_integrated_2019, sahin_filter_2017, zhang_modified_2017}.
	Under such a waveform reuse design, the ``radar signal'' at  communication received is not interference but the expected communication signal.
	
	In order to implement joint transmit beamforming, the precoders $\bm{W}_c$ and $\bm{W}_r$ are to be jointly designed in consideration of the system performance.
	The performance metrics of MIMO radar and multiuser MIMO communication are detailed in  Section~\ref{section3}.
	In practice, the precoders should satisfy some constraints representing the transmit hardware.
	Here, we require that the transmit waveform satisfies a per-antenna power constraint, namely, that the transmit power of each antenna is identical.
	The per-antenna power constraint settles with the common practice that radar waveforms should be transmitted with their maximal available power \cite{stoica_probing_2007}, and has also been applied in multi-antenna communication systems \cite{4203115, 5733456, loyka2017capacity}.
	We note that the per-antenna power constraints can be extended to represent other power-related limitations, such as total power constraints, according to the hardware requirements. 
	
	To formulate the power constraint, define the covariance of transmit waveform as
	\begin{equation} \label{eq:R}
	\bm{R} = \mathbf{E} \left(  \bm{x} [n] \bm{x} ^ H [n] \right)  .
	\end{equation}
	Substituting  \eqref{eq-1-1}-\eqref{eq-1-2} into \eqref{eq:R} yields the covariance $\bm{R}$ as
	\begin{equation}
	\bm{R} = \bm{W}_r \bm{W}_r ^ H + \bm{W}_c \bm{W}_c ^ H.
	\end{equation}
	The per-antenna power constraint implies that for each $\ m = 1,\ldots,M$ it holds that
	\begin{equation} \label{ch2-eq-power}
	\left[\bm{R}\right]_{m,m}  = \left[\bm{W}_r \bm{W}_r ^ H + \bm{W}_c \bm{W}_c ^ H\right]_{m,m}  =   P_t / M , 
	\end{equation}
	where $P_t$ is the total transmit power.    
	Under this constraint, we discuss the radar and communication metrics for precoder design in the following section.

	\section{Performance metrics of radar and communication} \label{section3}
	Based on the signal model of joint transmit beamforming, we aim to design the precoders in light of the following guidelines:
	For MIMO radar, the precoder is designed to synthesize transmit beams towards radar targets of interests; 
	For multiuser MIMO communication, the precoder is designed to guarantee the receiving SINR at communication users.
	These performance metrics of MIMO radar  and multiuser MIMO communication are properly formulated in Subsections~\ref{section3-1} and \ref{section3-2}, respectively.
	
	\subsection{MIMO Radar Performance} \label{section3-1}
	The main purpose of MIMO radar beamforming is to direct the transmit beam towards several given directions, so that one can obtain more information of the targets illuminated by these beams. These directions are typically known to the transmitter: 
	When radar works in tracking mode, the beam direction is inferred from the direction of the targets acquired at previous observations;
	When radar works in searching mode, the beam direction is given by the center of angular sector-of-interest. 
	Consequently, to formulate the performance metric associated with MIMO radar beamforming, we first express the transmitted signal at each direction, and then develop a loss function evaluating the transmit beam pattern. Combining the loss function and the per-antenna power constraint, we achieve an optimization problem which accounts for the radar performance.
	
	In DFRC systems, the communication signals can also be used for sensing, since the radar receiver has complete knowledge of the transmitted communication waveform. In this way, the communication signal is not regarded as interference at the radar receiver.
	Under the assumption that the transmit waveform is narrow-band and the propagation path is line of sight (LoS), the baseband signal at direction $\theta$ can be expressed as
	\begin{equation}
	y[n; \theta] = \bm{a} ^ H (\theta)  \bm{x} [n],
	\end{equation}
	where $\bm{a}(\theta)$ is the array steering vector of direction $\theta$.
	When the waveform is reflected from a point target located at angular direction $ \theta $, the received signal can be written as
	\begin{equation} \label{eq-radar-receive}
	\bm r [n] = \beta \bm a ^ c(\theta) \bm a ^ H (\theta) \bm x [n - n'] + \bm v [n],
	\end{equation} 
	where $ \beta $ is the complex amplitude proportional to the radar-cross sections (RCS) of the target, $n'$ represents the discrete time delay, and $ \bm v[n]$ is additive zero-mean temporally-white noise with covariance $\bm R _ v$. 
	
	Following the guidelines for MIMO radar probing signal design stated in \cite{stoica_probing_2007}, the desired goals of MIMO radar transmit beamforming include:
	\begin{enumerate}[1)]
		\item Optimize the transmit power at given directions, or generally match a desired beam pattern;
		\item Decrease the cross correlation pattern among signals at several given directions, which is essential for the performance of adaptive MIMO radar techniques.
	\end{enumerate}
	Here, the transmit power (beam pattern) at angular direction $\theta$ is
	\begin{align}  
	& P(\theta; \bm{R}) = \mathbf{E} \left( \left| y[n;\theta] \right|^2  \right) \notag \\
	&= \mathbf{E} \left( \bm{a} ^ H (\theta)  \bm{x} [n] \bm{x} ^ H [n] \bm{a} (\theta) \right)   = \bm{a} ^H (\theta)  \bm{R} \bm{a}(\theta),
	\label{eq-2-3}
	\end{align}  
	and the cross correlation pattern between direction $\theta_1$ and $\theta_2$ is defined as
	\begin{align}   
	& P_c (\theta_1, \theta_2; \bm{R}) = \mathbf{E} \left( y ^ *[ n; \theta_1 ]  y[n; \theta_2]  \right) \notag \\
	&= \mathbf{E} \left( \bm{a} ^ H (\theta_2)  \bm{x} [n] \bm{x} ^ H [n] \bm{a} (\theta_1) \right)   = \bm{a} ^H (\theta_2)  \bm{R} \bm{a}(\theta_1).
	\label{eq-2-3-1}
	\end{align}    
	From \eqref{eq-2-3} and \eqref{eq-2-3-1},  both the transmit beam pattern and cross correlation pattern are determined by the covariance $\bm{R}$.
	Then, properly beamforming of MIMO radar waveforms is achieved by  designing the covariance matrix $\bm R$ \cite{stoica_probing_2007, fuhrmann_transmit_2008}.
	
	To this aim, we use the loss function proposed in \cite{stoica_probing_2007,fuhrmann_transmit_2008}  to evaluate the radar performance, which is the weighted sum of two parts: beam pattern error and cross correlation. In particular, the first part is the mean square error (MSE) between the obtained beam pattern and some desired beam pattern,  given by
	\begin{equation}  \label{eq-2-7}
	L_{r,1}(\bm{R}, \alpha) =  \frac{1}{L} \sum _ {l =1} ^ L  \left| \alpha d(\theta_l) - P(\theta_l; \bm{R}) \right| ^ 2,
	\end{equation}
	where $\alpha$ is a scaling factor, $d(\theta)$ is the given desired beam pattern, and $\left\{\theta_l \right \} _{l=1} ^ {L}$ are sampled angle grids.
	The second part is the mean-squared cross correlation pattern, expressed as
	\begin{equation}  \label{eq-2-8}
	L_{r,2}(\bm{R}) = \frac{2 }{P^2 - P} \sum_{p = 1}^{P-1} \sum_{q = p+1}^{P} \left| P_c( \overline{ \theta} _p,  \overline{ \theta}  _q ; \bm{R})  \right| ^ 2,
	\end{equation}   
	where $\{ \overline{ \theta} _p   \} _{p=1} ^ {P}$ are the given  directions of the targets. The summation in \eqref{eq-2-8} is normalized by $ \frac{2 }{P^2 - P} $, as there exists $\frac{P^2 - P}{2}$ pairs of distinct directions in the set $\{ \overline{ \theta} _p   \}$. 
	The loss function of radar is  then
	\begin{equation} \label{eq-2-9}
	L_{r}(\bm{R}, \alpha) = L_{r,1}(\bm{R}, \alpha) + w_c L_{r,2}(\bm{R} ), 
	\end{equation}
	where $w_c$ is a weighting factor.
	As discussed in \cite{stoica_probing_2007,fuhrmann_transmit_2008}, the loss function $L_{r}(\bm{R}, \alpha)$ can be written as a positive-semidefinite quadratic function of $\bm{R}$ and $\alpha$. 
	
	Combining the loss functions in \eqref{eq-2-7} and \eqref{eq-2-8}, the  covariance of the transmitted signal in the absence of communication constrains, i.e., in a radar-only setup, can be designed in light of the overall radar objective under per-antenna power constraints  \cite{stoica_probing_2007}, i.e. 
	\begin{subequations} \label{eq:radaropt}
		\begin{align}
		\min_{\bm{R} , \alpha }  \  \ \ \ & \  L_r(\bm{R}, \alpha )  \\
		\mathrm{subject \ to} \  & \ \bm{R}   \in \mathcal{S}_{M}^+, \\ 
		& \ \left[\bm{R}\right]_{m,m} = P_t / M , \ m = 1,\ldots,M.
		\end{align}
	\end{subequations} 
	We denote the optimal covariance of this problem by $\bm{R}_0$.
	Generally, the  performance requirement of multiuser MIMO communication, detailed in Subsection~\ref{section3-2}, cannot be satisfied if the covariance of transmit waveform is $\bm{R}_0$.  
	In other words, there is an inherent radar performance loss due to spectrum sharing with communications compared to the radar-only case.
	To address the communication performance of our DFRC system, we discuss the communication metric in the next subsection.

	\subsection{Multiuser MIMO Communication Performance} \label{section3-2}
	A common performance measure for multiuser broadcast communications is the SINR, which is directly related to the achievable rate under reduced complexity decoding \cite[Ch. 8]{el2011network}. Consequently, we design the precoders of  MIMO transmission to optimize the  users' SINR. To this aim, we first present the communication signal model, and derive the expression of SINR with respect to the precoders (for radar and communication signals) and the channel matrix, followed by a formulation of the combined objective which accounts for communication performance. 
	
	Consider a down-link multiuser MIMO transmission scenario with $K<M$ single antenna users observing the output of a frequency flat Gaussian noise channel.    The channel output at the $K$ users at time instance $n$, represented via the $K \times 1$ vector ${\bm r}[n]$, is given by
	\begin{equation} \label{eq-comm-receiver}
	\bm{r} [n] = \bm{H} \bm{W}_c \bm{c} [n] + \bm{H} \bm{W}_r \bm{s} [n] + \bm{v} [n], 
	\end{equation} 
	where $\bm{H}$ is the $K\times M$ narrow-band channel matrix and $\bm{v}[n]$ is additive white Gaussian noise (AWGN) with covariance $\sigma^2 \bm{I}_K$.

	In multiuser transmit beamforming, the precoder should be designed to guarantee a certain level of SINR at the users.
	Here, it is assumed that the transmit array knows the instantaneous downlink channel $\bm H$. This knowledge can be obtained for example, by exploiting wireless channel reciprocity when operating in  time-division duplex mode, i.e., the downlink channel is obtained via uplink channel estimation. Alternatively, in frequency-division duplex mode, downlink channel can be obtained via channel feedback from the users, see, e.g., \cite{4299617}.
	Define the equivalent radar-to-user channel and equivalent inter-user channel  matrices as 
	\begin{subequations}\label{eq:FrFc}
		\begin{align}
		\bm{F}_r = \bm{H} \bm{W}_r,\\
		\bm{F}_c = \bm{H} \bm{W}_c,
		\end{align}
	\end{subequations}
	respectively.
	Since the users are generally not able to cooperate with each other, the off-diagonal elements of $\bm{F}_c $ lead to inter-user interference, which should be mitigated by precoding.
	At the same time, since the users generally do not have any prior information on radar waveform, $\bm{F}_r $ leads to interference from radar.  
	At the $k$-th user, the signal power is 
	\begin{equation}
	\mathbf{E} \left(  | [ \bm{F}_c ]_{k,k} c_k(t) | ^ 2 \right)  = \left|  [ \bm{F}_c ]_{k,k} \right | ^ 2,
	\end{equation}
	the power of inter-user interference is 
	\begin{equation}
	\mathbf{E} \Big( \sum_{i \neq k}  | [ \bm{F}_c ]_{k,i} c_i(t)  ^ 2 \Big)  =   \sum_{i \neq k} \left |   [ \bm{F}_c ]_{k,i} \right|^2,
	\end{equation}
	and the power of interference from radar is
	\begin{equation}
	\mathbf{E} \Big( \sum_{i = 1}^{M}  | [ \bm{F}_r ]_{k,i} s_i(t) | ^ 2 \Big)  =   \sum_{i = 1}^{M}   \left| [ \bm{F}_r ]_{k,i}  \right|^2.
	\end{equation}
	Therefore, the SINR at the $k$-th user is expressed as
	\begin{equation} \label{eq-ch3-sinr}
	\gamma_k = \frac{| [ \bm{F}_c ]_{k,k} | ^ 2 }{\sum_{i \neq k}  | [ \bm{F}_c ]_{k,i} |^2 + \sum_{i = 1}^{M}  | [ \bm{F}_r ]_{k,i} |^2 + \sigma ^ 2 }.
	\end{equation}

	Two typical design criteria for multiuser beamforming are \cite{liu_coordinated_2011, wiesel_zero-forcing_2008}:
	\begin{itemize}
		\item \emph{Throughput}: maximizing the sum rate
		\begin{equation}
		\label{eq-sum-rate}
		C(\bm{\gamma}) = \sum_{k=1}^{K} \log_2 \left(1 + \gamma_k \right),
		\end{equation} 
		\item \emph{Fairness}: maximizing the minimal SINR, referred to henceforth as the fairness SINR: 
		\begin{equation}
		\label{eq-fainess-sinr}
		F(\bm{\gamma}) = \min \{\gamma_1, \ldots, \gamma_K \},
		\end{equation}
	\end{itemize} 
	where $\bm{\gamma} = [\gamma_1, \ldots, \gamma_K] ^ T$.
	In this work, we use fairness SINR $ F(\bm{\gamma}) $ as the performance metric for multiuser communication, and require it to be higher than a given threshold $\Gamma$, guaranteeing a minimal level of communication quality of  service at each user, i.e.
	\begin{equation}\label{eq:sinrthreshold}
	\gamma_ k \geq \Gamma, \ k = 1,\ldots,K.
	\end{equation}
	Moreover, the fairness beamforming is  simpler in terms of computation complexity and  can be solved in polynomial time, while the optimal throughput beamforming is NP hard \cite{liu_coordinated_2011}.
	We note that in the formulated joint beamforming problem in Section~\ref{section4}, the fairness SINR requirement can be extended to having $K$  individual SINR constraints \cite{schubert_solution_2004},  namely,
	\begin{equation}\label{eq:sinrthreshold2}
	\gamma_ k \geq \Gamma_k, \ k = 1,\ldots,K,
	\end{equation}
	where $\Gamma_k$ is the SINR threshold at the $k$-th user.

	
	\section{Joint Transmit Beamforming} \label{section4}
	With the proposed MIMO radar and communication performance metrics, we now turn to design a DFRC joint beamforming scheme. We begin by formulating the joint transmit beamforming as an optimization problem with respect to the precoding matrices in Subsection~\ref{subsec:optformulation}. 
	To solve this problem, we  propose a semidefinite relaxation (SDR) based optimization scheme in \ref{subsec:SDR}, and a zero-forcing (ZF) methods which cancels the inter-user interference and the radar interference in Subsections~\ref{sub:zf:FcFr}, respectively.
	
	\subsection{Problem formulation}\label{subsec:optformulation}
	The goal of our joint DFRC beamforming  is to optimize the radar beam pattern under the transmit power and communication quality of service constraints. In particular, we minimize the loss function on radar beam pattern defined in \eqref{eq-2-9}, under the per-antenna power constraint \eqref{ch2-eq-power}  and the fairness SINR constraint \eqref{eq:sinrthreshold} for each downlink user.
	
	Let $\bm{W} = [\bm{W}_c, \bm{W}_r] $ be the overall precoding matrix.
	The precoding matrix can be obtained by solving the following optimization problem
	\begin{subequations} \label{eq-3-22}
		\begin{align}
		\min_{\bm{W}, \alpha}  \ & \  L_r(\bm{R}, \alpha )  \label{eq:costfunction}\\ 
		\mathrm{subject \ to} \  & \ \bm{R} = \bm{W}\bm{W}^ H   \in \mathcal{S}_{M}^+, \label{eq-3-22b} \\
		& \ \left[\bm{R}\right]_{m,m} = P_t / M , \ m = 1,\ldots,M , \label{eq:powerconstraint} \\ 
		& \ \gamma_ k \geq \Gamma, \  k = 1,\ldots, K, \label{eq-3-22e}
		\end{align}
	\end{subequations}  
	where \eqref{eq:costfunction}-\eqref{eq:powerconstraint} come from \eqref{eq:radaropt} addressing the radar performance, and \eqref{eq-3-22e} follows from considering the fairness SINR requirement  \eqref{eq:sinrthreshold}.
	
	The selection of the threshold $\Gamma$ affects the trade-off between the communication quality and radar performance. 
	When $\Gamma = 0$,  \eqref{eq-3-22e} always holds, and the joint radar-communication beamforming problem \eqref{eq-3-22} reduces to the radar-only optimization \eqref{eq:radaropt}.
	When $\Gamma > 0$, compared with the radar-only transmit beamforming problem in \eqref{eq:radaropt}, the precoder $\bm W$, which dictates the equivalent channels via \eqref{eq:FrFc}, is restricted by the SINR constraints in \eqref{eq-3-22e}.
	Therefore, there can be an inherent radar performance loss  induced by the need to meet the communication performance guarantees, as compared to the radar-only case.
	If higher $\Gamma$ is set, higher signal power and less interference is expected  to be observed at the user side, further restricting the precoding matrices.
	As a result, the performance loss of MIMO radar becomes more significant if  higher $\Gamma$ is set.
	
	The optimization problem \eqref{eq-3-22} is not convex because of the quadratic equality constraint in \eqref{eq-3-22b} and is thus difficult to solve. 
	Nonetheless, we show in Subsection~\ref{subsec:SDR} that it can be recast using semidefinite relaxation (SDR) such that the solution to the solvable relaxed problem is also the global optimizer of the original non-convex \eqref{eq-3-22}, i.e., the relaxation is tight. 
	To further reduce the computation complexity, we propose a sub-optimal zero-forcing beamforming strategy in Subsection~\ref{sub:zf:FcFr}, which is shown to be able to approach the performance of the global solution to \eqref{eq-3-22}   in our numerical study presented in Section~\ref{section5}.

	\subsection{Joint Transmit Beamforming via SDR}
	\label{subsec:SDR}
	In this subsection, we tackle the non-convex problem \eqref{eq-3-22} using an SDR strategy \cite{luo_semidefinite_2010,palomar2009convex}. 
	To this aim, we first explicitly write the  relationship \eqref{eq-3-22b} as a  quadratic constraint with respect to each column of $\bm W$. 
	Let $\bm w _ i$ denote the $i$-th column of $\bm W$, for $i = 1, \ldots, M + K$. 
	Then  \eqref{eq-3-22b}  becomes
	\begin{equation}
	\label{eqn:RasSum}
	\bm{R} = \sum_{i=1}^{M+K} \bm w _ i \bm w _ i ^ H.
	\end{equation}    
	Defining $\bm{R}_i = \bm{w}_i \bm {w}_i ^ H$, we have
	\begin{equation}
	\bm{R} = \sum_{i=1}^{M+K} \bm R _ i,
	\end{equation}
	where we omit the rank-one constraints.
	The SINR constraints in \eqref{eq-3-22e} can  be converted to linear constraints in the rank-one matrices $\{ \bm{R}_{i} \}$.
	Letting $\bm h _ {k} ^ H$ denote the $k$-th row of $\bm H$,  $k = 1,\ldots,K$, the entires of the equivalent channel matrices can be written as $ [ \bm F _ c]_{k,i} = \bm h _ k ^ H \bm w _ i$, and $ [ \bm F _ r]_{k,i} = \bm h _ k ^ H \bm w _ {i+K}$. 
	Consequently, the SINR constriant becomes
	\begin{align} 
	\gamma_k &= \frac{  \bm h _ k ^ H \bm w _ k \bm w _ k ^ H \bm h _ k }{ \sum_{1\leq i \leq M+K, i \neq k} \bm h _ k ^ H \bm w _ i \bm w _ i ^ H \bm h _ k + \sigma^2} \notag \\
	&= \frac{  \bm h _ k ^ H \bm R _ k  \bm h _ k }{ \sum_{1\leq i \leq M+K, i \neq k} \bm h _ k ^ H \bm R _ i \bm h _ k + \sigma^2} \notag \\
	&= \frac{  \bm h _ k ^ H \bm R _ k  \bm h _ k }{   \bm h _ k ^ H \bm R   \bm h _ k - \bm h _ k ^ H \bm R _ k \bm h _ k + \sigma^2} \geq \Gamma.
	\label{eq-sinr-sdr}
	\end{align}
	
	We now cast \eqref{eq-3-22} as an equivalent  quadratic semidefinite programming (QSDP) with rank-one constraints
	\begin{subequations} \label{eq-3-rank}
		\begin{align}
		\min_{ \bm{R}. \{ \bm{R} _i \},  \alpha}  \ & \  L_r(\bm{R}, \alpha)  \\
		\mathrm{subject \ to} \  & \ \bm{R}  = \sum_{i=1}^{M+K} \bm R _ i   \in \mathcal{S}_{M}^+ , \label{eq-3-rankb} \\
		& \ \left[\bm{R}\right]_{m,m} = P_t / M , \ m = 1,\ldots,M ,  \\
		& \ \bm{R}_i \in \mathcal{S}_{M}^+, \ \mathrm{rank} (\bm R _ i) = 1 , \ i = 1,\ldots,K+M,
		\label{eq-3-rankd}  \\
		\ \left(1 + \Gamma^{-1} \right) &\bm h _ k ^ H \bm R _ k  \bm h _ k  \geq \bm h _ k ^ H \bm R   \bm h _ k + \sigma^2 , \ k = 1,\ldots,K, \label{eq-3-ranke}
		\end{align}
	\end{subequations}
	where $\eqref{eq-3-ranke}$ is derived from \eqref{eq-sinr-sdr}.
	We observe that in problem \eqref{eq-3-rank}, the individual matrices $ \{ \bm R _ i \}_{i \geq K+1}$ have no effect on the SINR constraints and are only encapsulated in the overall covariance matrix $\bm R$.
	Therefore, we can remove the variables $ \{ \bm R _ i \}_{i \geq K+1}$, and \eqref{eq-3-rank} is relaxed to 
	\begin{subequations} \label{eq-4-rank}
		\begin{align}
		\min_{ \bm{R}, \bm{R}_1, \ldots, \bm{R}_K,  \alpha}  \ & \  L_r(\bm{R}, \alpha)  \\
		\mathrm{subject \ to} \  & \ \bm{R} \in \mathcal{S}_{M}^+, \ \bm{R} - \sum_{k=1}^{K} \bm R _ k   \in \mathcal{S}_{M}^+ , \label{eq-4-rankb} \\
		& \ \left[\bm{R}\right]_{m,m} = P_t / M , \ m = 1,\ldots,M ,  \\
		& \ \bm{R}_k \in \mathcal{S}_{M}^+, \ \mathrm{rank} (\bm R _ k) = 1 , \ k = 1,\ldots,K,
		\label{eq-4-rankd}  \\
		\ \left(1 + \Gamma^{-1} \right) &\bm h _ k ^ H \bm R _ k  \bm h _ k  \geq \bm h _ k ^ H \bm R   \bm h _ k + \sigma^2 , \ k = 1,\ldots,K. \label{eq-4-ranke}
		\end{align}
	\end{subequations}
	
	The optimization problem \eqref{eq-4-rank} is still non-convex because of the rank-one constraints.
	Omitting these constraints leads to the following relaxation:
	\begin{subequations} \label{eq-5-rank}
		\begin{align}
		\min_{ \bm{R}, \bm{R}_1, \ldots, \bm{R}_K,  \alpha}  \ & \  L_r(\bm{R}, \alpha) \label{eq:sdr:objectives}  \\
		\mathrm{subject \ to} \  & \ \bm{R} \in \mathcal{S}_{M}^+, \ \bm{R} - \sum_{k=1}^{K} \bm R _ k   \in \mathcal{S}_{M}^+ , \label{eq-5-rankb} \\
		& \ \left[\bm{R}\right]_{m,m} = P_t / M , \ m = 1,\ldots,M ,  \label{eq-5-rankc}  \\
		& \ \bm{R}_k \in \mathcal{S}_{M}^+, \ k = 1,\ldots,K,
		\label{eq-5-rankd}  \\
		\ \left(1 + \Gamma^{-1} \right) & \bm h _ k ^ H \bm R _ k  \bm h _ k  \geq \bm h _ k ^ H \bm R   \bm h _ k + \sigma^2 , \ k = 1,\ldots,K. \label{eq-5-ranke}
		\end{align}
	\end{subequations}    
	This relaxed optimization model  \eqref{eq-5-rank} is a convex  QSQP,  because the target function is a positive-semidefinite quadratic form and all the
	constraints are either linear or semidefinite. 
	The global optimum of \eqref{eq-5-rank} can be obtained in polynomial time  with convex optimization toolboxes \cite{ vandenberghe1996semidefinite,toh1999sdpt3, tutuncu2003solving , Toh2012}.

	The relaxation used in SDR is tight if the optimal  $\bm R_1, \ldots,\bm R _ K $ for \eqref{eq-5-rank} are exactly rank-one, i.e., the solution to the relaxed problem is also a solution to the original non-convex problem. While such relaxations are not necessarily tight, the SDR used in obtaining \eqref{eq-5-rank} from \eqref{eq-4-rank} is tight, as stated in the following theorem:
	\begin{mythm} \label{thm1}
		There exists a global optimum for \eqref{eq-5-rank}, denoted by  $ \tilde{ \bm R },\tilde{ \bm{R} }_{1}, \ldots, \tilde{ \bm{R}}_{K}$,   satisfying
		\begin{displaymath}
		\mathrm{rank} (\tilde{ \bm{R} }_k) = 1, \ k = 1,\ldots,K. 
		\end{displaymath}
	\end{mythm}
	
	\begin{proof}
		See Appendix \ref{app:Proof1}.
	\end{proof}
	
	We note that Theorem~\ref{thm1} only states that the rank-one global optimum exists. Generally, the global optimum to \eqref{eq-5-rank} may not be unique and convex optimization software may not give a rank-one solution.
	Once the optimal solution $ \hat{\bm{R}}, \hat{\bm{R}}_1,\ldots,\hat{\bm {R}} _ K$  are obtained, we use them to obtain the  rank-one optimal solution 
	$\tilde{\bm{R}}_1,\ldots,\tilde{\bm {R}} _ K$ and the  corresponding optimal precoder $\tilde{\bm{w}}_1,\ldots,\tilde{\bm {w}} _ K$, as presented in Appendix \ref{app:Proof1}.
	First, we compute $\tilde{\bm{R}}_1,\ldots,\tilde{\bm {R}} _ K$ and $\tilde{\bm R}, \tilde{\bm{w}}_1,\ldots,\tilde{\bm {w}} _ K$ via
	\begin{equation} \label{eq-w1-k}
\tilde{\bm R}=  \hat{\bm R}, \	\tilde{ \bm w }_k = \big(\bm h_k ^ H \hat{ \bm R } _ k \bm h _ k \big) ^ {-1/2} \hat { \bm R } _ k \bm h _ k, \ \tilde{ \bm R }_k = \tilde{ \bm w }_k \tilde{ \bm w }_k ^ H,  
	\end{equation}
	for  $k = 1,\ldots,K$.
	According to the proof of Theorem~\ref{thm1}, $\tilde{\bm{R}}, \tilde{\bm{R}}_1,\ldots,\tilde{\bm {R}} _ K$ is optimal to \eqref{eq-4-rank} and hence is also optimal \eqref{eq-5-rank}.
	To show that  $\tilde{\bm{R}}, \tilde{\bm{R}}_1,\ldots,\tilde{\bm {R}} _ K$ is also  optimal to \eqref{eq-3-rank}, we construct rank-one matrices $ \{ \tilde{ \bm R} _ i \}_{i \geq K+1}$   as $\tilde{\bm{R}}_i= \tilde{\bm{w}} _ i\tilde{\bm{w}} _ i ^ H$, where the vectors $\tilde {{\bm w}}_i$ for $i > K$ are calculated by the Cholesky decomposition \cite{zhang2017matrix}
	\begin{equation} \label{eq-wk-m}
	\bm W _ r \bm W _ r ^ H = \tilde{ \bm{R}} - \sum_{k=1}^{K} \tilde{ \bm w} _ k \tilde{\bm {w}} _ k ^ H ,
	\end{equation}
	where $\bm  W _ r = \left[\tilde{ \bm w} _ {K+1}, \ldots, \tilde{\bm w}_{K+M} \right]$ is a lower triangular matrix.
	From \eqref{eq-wk-m}, it can be verified that constraint \eqref{eq-3-rankb} holds for $\tilde{\bm R}, \tilde{\bm R}_1, \ldots, \tilde{\bm R}_{K+M}$.
	Therefore,  $\tilde{\bm R}, \tilde{\bm R}_1, \ldots, \tilde{\bm R}_{K+M}$ is a feasible solution to \eqref{eq-3-rank} and hence is  also an optimal solution to \eqref{eq-3-rank}.
	Furthermore, the precoding matrix $\tilde{\bm W} = \left[ \tilde{ \bm w }_1,\ldots, \tilde{ \bm w }_{K+M} \right]$ is a solution to \eqref{eq-3-22}.
	
	We summarize the procedure to compute the precoding  matrix $\bm{W} $ in Algorithm~\ref{A-SDR}.
	The main computational burden in Algorithm~\ref{A-SDR} stems from solving the QSDP \eqref{eq-5-rank}.
	Specifically, given a solution accuracy $\epsilon$, the worst case complexity to solve the  QSDP  \eqref{eq-5-rank} with the primal-dual interior-point algorithm
	in \cite{nie_predictorcorrector_2001, toh_inexact_2008 } is $ \mathcal{O} (K ^ {6.5} M ^ {6.5} \log (1 / \epsilon) )$.
	%
	
	\begin{algorithm}
		
		\begin{algorithmic}[1]
			\Require  
			\Statex \hspace{-1.5em} Total transmit power $P_t$;
			\Statex \hspace{-1.5em}  Power of AWGN at users $\sigma^2$;
			\Statex \hspace{-1.5em} Expression of the MIMO radar loss function $L_r(\bm{R}, \bm{\alpha})$;
			\Statex \hspace{-1.5em} Instantaneous downlink channel $\bm{H}$;
			\Statex \hspace{-1.5em} SINR threshold $\Gamma$.
			\Ensure 
			\Statex \hspace{-1.5em} The overall precoding matrix $\bm{W}$.	
			\Statex{ \hspace{-2em} \textbf{Steps: } }
			\State   Compute the optimal value of $\hat{\bm{R}},\hat{\bm{R}}_1,\ldots,\hat{\bm{R}}_K$  by solving  \eqref{eq-5-rank} with convex optimization solvers. 
			\State Compute $\tilde{\bm w}_1, \ldots, \tilde{\bm w} _ K$ via \eqref{eq-w1-k}.
			\State Compute $\tilde{\bm w}_{K+1}, \ldots,\tilde{ \bm w }_ {K+M}$ via \eqref{eq-wk-m}.
			\State Set the overall precoding matrix $  \tilde{\bm W} = \left[ \tilde{ \bm w }_1,\ldots, \tilde{ \bm w }_{K+M} \right]$.

			\caption{Joint transmit beamforming via SDR}
			\label{A-SDR}
		\end{algorithmic}
	\end{algorithm}

	\subsection{Joint Transmit Beamforming  via ZF}
	\label{sub:zf:FcFr}
	The computational burden associated with obtaining the precoder via Algorithm~\ref{A-SDR} motivates seeking a reduced complexity sub-optimal beamforming strategy. 
	In this subsection, we focus on ZF beamforming. ZF methods facilitate obtaining closed-form, tractable, and interpretable precoders \cite{dimic_downlink_2005,wiesel_zero-forcing_2008}. In addition to its relative simplicity, from a communications perspective, ZF beamforming is known to asymptotically approach the sum-capacity in  broadcast channels \cite{yoo_optimality_2006}, indicating its potential to approach optimal performance in setups involving multi-user communications.
	
	We design the precoders to  eliminate the inter-user interference and radar interference, obtained by restricting $\bm F_c$ to a diagonal matrix and $\bm F_ r $ to a zero matrix, i.e.
	\begin{equation}\label{eq:Fr=0}
	\bm{F}_c = \mathrm{diag} \left( \sqrt{p_1},\ldots,\sqrt{p_K} \right), \  \bm F_r = \bm 0_{K \times M}.
	\end{equation}
	Here, $p_k$ is the signal power at the $k$-th user, for $1 \leq k \leq K$.
	Enforcing the interference to be canceled facilitates achieving high SINR values at the users. In our numerical study in Section~\ref{section5} we demonstrate that the achievable performance under the additional ZF constraint approaches that of the global solution to \eqref{eq-3-22}, obtained with increased computational burden via Algorithm~\ref{A-SDR},  when the SINR threshold is high.
	
	
	In ZF beamforming, the SINR constraint \eqref{eq-3-22e} is reformulated as
	$\frac{1}{\Gamma} p_k \geq  \sigma ^ 2$, 
	and the corresponding optimization problem \eqref{eq-3-22} becomes
	\begin{subequations} \label{eq-3-23}
		\begin{align}
		\min_{\bm{W},   \alpha}  \ & \  L_r(\bm{R},  \alpha )  \\
		\mathrm{subject \ to} \  & \ \bm{R} = \bm{W}\bm{W}^ H   \in \mathcal{S}_{M}^+,  \label{eq-3-23b}  \\
		& \ \left[\bm{R}\right]_{m,m} = P_t / M , \ m = 1,\ldots,M ,  \\
		& \ \bm{H W} =  \left[ \mathrm{diag}( \sqrt{\bm{p}}),\bm{0}_{K \times M} \right]
		, \label{eq-3-23d} \\
		& \ \frac{1}{\Gamma} {p}_k \geq  \sigma ^ 2, \ k = 1,\ldots,K. \label{eq:optII:nonconvex:SINRconstraint}
		\end{align}
	\end{subequations} 
	
	The ZF beamforming optimization \eqref{eq-3-23} is still non-convex.
	The following theorem shows that it can be converted to convex problem. 
	\begin{mythm}   \label{thm2}    
		Given a covariance matrix $\bm{R} \in \mathcal{S}_n ^ +$ and a full rank $K \times (K+M)$ matrix $\bm F$, there exists a matrix $\bm{W}$ satisfying \eqref{eq-3-23b} and 
		\begin{equation} \label{eq-3-23d1}
		\bm H \bm W = \bm F         
		\end{equation}
		if and only if 
		\begin{equation} \label{eq-3-24}
		\bm{H} \bm{R} \bm{H} ^ H = \bm{F} \bm{F} ^ H .
		\end{equation} 
	\end{mythm}
	\begin{proof}
		See Appendix \ref{app:Proof2}.
	\end{proof}
	
	Theorem~\ref{thm2} indicates that constraints \eqref{eq-3-23b} and \eqref{eq-3-23d} are equivalent to 
	\begin{equation}\label{eq:HRH}
	\bm{H} \bm{R} \bm{H} ^ H = \mathrm{diag} \left(\bm{p}\right),
	\end{equation}
	by letting $ \bm F =  \left[ \mathrm{diag}( \sqrt{\bm{p}}),\bm{0}_{K \times M} \right] $. 
	Using  \eqref{eq:HRH}, the globally optimal $\bm R$ to \eqref{eq-3-23} is found by  
	\begin{subequations} \label{eq-3-34}
		\begin{align}
		\min_{\bm{R}, \alpha }  \ & \  L_r(\bm{R}, \alpha ) \label{eq:zf-iu:objectives} \\
		\mathrm{subject \ to} \  & \ \bm{R}  \in \mathcal{S}_+ ^ M , \ \bm{H R H}^ H =  \mathrm{diag} \left(\bm{p}\right),  \label{eq-3-34b} \\
		& \ \left[\bm{R}\right]_{m,m} = P_t / M , \ m = 1,\ldots,M ,  \\
		& \ \frac{1}{\Gamma} {p}_k \geq  \sigma ^ 2, \  k = 1,\ldots,K. \label{eq:optII:convex:SINRconstraint}
		\end{align}
	\end{subequations} 
	Similar to \eqref{eq-5-rank}, the optimization   \eqref{eq-3-34} is a  convex QSDP, and the global optimum of \eqref{eq-3-34} can be obtained in polynomial time. As we show in the sequel, the overall complexity of ZF beamforming is substantially lower than that of recovering the global optimum via Algorithm~\ref{A-SDR}.
	
	The solution of \eqref{eq-3-34}, i.e., the matrix $\tilde {\bm R}$ and the vector $\tilde {\bm p}$, are used to construct the optimal precoding matrix $\tilde {\bm W}$, As detained in the proof of  Theorem~\ref{thm2}. 
	Here, we briefly give the final expressions.  
	First, we recover an $M\times M$ matrix $\bm L_r$  which satisfies $\tilde {\bm R}= \bm L_r \bm L_r^H$. This can be obtained using, e.g.,  Cholesky decomposition, though $\bm{L}_r$ does not have to be triangular and any matrix satisfying $\tilde {\bm R}= \bm L _r \bm L_r^H$ may be used to calculate $\tilde {\bm W}$.
	Then, the resulting precoder $\tilde {\bm R}$ is
	\begin{equation}\label{eq:WLQhQf}
	\tilde {\bm W} = \bm L_r \bm Q_h^H \left[\bm{Q}_f^T\right]^T_{{1:M}},
	\end{equation}
	where 
	$\bm Q_h$ and $\bm Q_f$ are obtained by applying row QR decomposition to $\bm H \bm L_r$ and $\bm F$, respectively. 
	Since  $\bm{F} =  \left[ \mathrm{diag} (\sqrt{\tilde{\bm{p}}}), \bm{0}_{K \times M} \right]$ is diagonal, it holds that $\bm Q_f = \bm I_{M+K}$, and thus \eqref{eq:WLQhQf} is simplified to
	\begin{equation} \label{eq:WLrQh} 
	\tilde {\bm W}
	=  \left[ \bm L_r \bm Q_h^H , \bm 0_{M \times K} \right]. 
	\end{equation}
	According to the proof of Theorem~\ref{thm1}, $\tilde {\bm W}, \tilde {\bm R}$ is a feasible solution to \eqref{eq-3-23}. Since $\tilde {\bm R}$ is the global optimum to \eqref{eq-3-23}, $\tilde {\bm W}$ is also globally optimal to \eqref{eq-3-23}. 
	

	The resulting ZF beamforming method is summarized below as Algorithm~\ref{A1}.
	The main computational burden in Algorithm~\ref{A1} stems from solving the QSDP problem \eqref{eq-3-34}, as is also the case in Algorithm~\ref{A-SDR}.
	Given a solution accuracy $\epsilon$, the worst case complexity to solve the  QSDP problem \eqref{eq-3-34} with the primal-dual interior-point algorithm   in \cite{nie_predictorcorrector_2001, toh_inexact_2008 } is $ \mathcal{O} ( M ^ {6.5} \log (1 / \epsilon) )$.
	Compared to the recovering the global solution via the SDR-based Algorithm~\ref{A-SDR}, the worst-case computation complexity for ZF beamforming is lower by a factor of $K^{6.5}$. This computational complexity reduction stems from the fact that the optimization problem  \eqref{eq-3-34} involves only one semidefinite constraint, while the problem \eqref{eq-5-rank}, from which Algorithm~\ref{A-SDR} originates, involves $K+2 = \mathcal{O} (K)$ such constraints.

	\begin{algorithm}
		
		\begin{algorithmic}[1]
			\Require 
			\Statex \hspace{-1.5em} Total transmit power $P_t$;      
			\Statex \hspace{-1.5em}   Power of AWGN at users $\sigma^2$;            
			\Statex \hspace{-1.5em}   Expression of the MIMO radar loss function $L_r(\bm{R}, \bm{\alpha})$;            
			\Statex \hspace{-1.5em}   Instantaneous downlink channel $\bm{H}$;           
			\Statex \hspace{-1.5em}  SINR threshold $\Gamma$.	
			\Ensure 
			\Statex \hspace{-1.5em} The overall precoding matrix $\tilde {\bm W}$.
			\Statex{ \hspace{-2em} \textbf{Steps: } }
			\State   Compute the optimal $\tilde {\bm R}$, and $\tilde {\bm p}$ by solving optimization problem \eqref{eq-3-34} with convex optimization solvers. 
			\State Compute the Cholesky decomposition of $\tilde {\bm R}$ as $ \tilde {\bm R} = \bm{L}_r \bm{L}_r ^ H$.
			\State Given $\bm{H L}_r$, calculate $\bm Q_h$ with the row QR decomposition \eqref{eq-3-26} shown  later in  Appendix \ref{app:Proof2}.
			\State Compute the overall precoding matrix  $\tilde {\bm W}$ using \eqref{eq:WLrQh}.

			\caption{Joint transmit beamforming via ZF}
			\label{A1}
		\end{algorithmic}
	\end{algorithm}

	We next discuss how the selection of the SINR threshold $\Gamma$ affects trade-off between communications and radar when using ZF beamforming. As noted in the discussion following the original optimization problem \eqref{eq-3-22}, the radar loss function here decreases as $\Gamma$ increases, i.e., the less restrictive the communication constraints are, the better the radar functionality can perform.
	However, there are two phenomenons which are explained in the sequel, that are different under ZF beamforming compared to the original optimization problem \eqref{eq-3-22}: 1) When $\Gamma$ approaches zero, the radar performance achieved  is generally different from the radar-only optimal performance; 2) The radar loss function and the obtained fairness SINR remain constant  if $\Gamma$ is lower than some positive value. 

	To understand phenomenon 1), we specialize the ZF optimization problem \eqref{eq-3-34} for the case of $ \Gamma = 0 $, resulting in
	\begin{subequations} \label{eq-3-37}
		\begin{align}
		\min_{\bm{R}, \bm{p} ,  \alpha}  \ & \  L_r(\bm{R}, \alpha )  \\
		\mathrm{subject \ to} \  & \ \bm{R}   \in \mathcal{S}_{M}^+, \\
		& \ \left[\bm{R}\right]_{m,m} = P_t / M , \ m = 1,\ldots,M,   \\
		& \ \bm{H R H}^ H =  \mathrm{diag} \left(\bm{p}\right). \label{eq:diagHRH}
		\end{align}
	\end{subequations} 
	Here, we note that this formulation is  distinct from the radar-only optimization problem \eqref{eq:radaropt}, since, even when the SINR can take any value, we still force the interference to be cancelled. This restriction is reflected in the  additional constraint \eqref{eq:diagHRH}  imposed on $\bm{R}$, namely that $\bm{H} \bm{R} \bm{H} ^ H$ should be a diagonal matrix.
	The optimal radar-only covariance $\bm{R}_0$, which is not forced to satisfy this interference cancelling constraint, generally does not satisfy it, i.e.  $\bm{H} \bm{R}_0 \bm{H} ^ H$ is not a diagonal matrix.
	If $\bm{R}_0$ is not a feasible solution of problem \eqref{eq-3-37}, the radar-only optimal performance  cannot be achieved. 
	
	In order to explain phenomenon 2), we again focus on the ZF optimization specialized to the case of no SINR constraints in  \eqref{eq-3-37}, and denote its solution by $\{ \bm{R}_{\rm II}, \alpha_{\rm II},  \bm {p}_{\rm II}\}$. 
	Given $\bm {p}_{\rm II}$, the resulting fairness SINR is given by 
	\begin{equation}
	\Gamma_{\rm II} = \min\{ \bm{p}_{\rm II} \}/ \sigma ^ 2.
	\end{equation}
	In  problem \eqref{eq-3-34}, if the given $\Gamma$ is not greater than $ \Gamma_{\rm II}$, i.e. $0 \le \Gamma \le \Gamma_{\rm II}$, the constraint \eqref{eq:optII:convex:SINRconstraint} always holds and can be regarded as being invariant to the actual solution $\{ \bm{R}_{\rm II}, \alpha_{\rm II},  \bm {p}_{\rm II}\}$.
	In this case, $\{ \bm{R}_{\rm II}, \alpha_{\rm II},  \bm {p}_{\rm II}\}$ is still a feasible solution for \eqref{eq-3-34}. Thus, the minimized radar loss function is  equal to $L_r(\bm{R}_{\rm II}, \alpha_{\rm II})$, and the obtained fairness SINR is equal to $\Gamma_{\rm II}$, for ZF beamforming derived under every SINR constraint  satisfying $\Gamma \le \Gamma_{\rm II}$.
	
	Here we compare the two proposed  beamforming methods.     
	The key difference between them is whether to completely eliminate the interference.
	As a globally optimal method, the radar performance of SDR beamforming should be better than that of the sub-optimal ZF beamforming under the same communication requirement.
	However, the performance gap may become small with reasonably large $\Gamma$, because the interference is expected to be eliminated under strict constraint on SINR. In this case, ZF beamforming is  preferable since its corresponding QSDP problem has a much simpler form.
	To explain the performance gap when the given $\Gamma$ is low, we note that it is unnecessary to completely eliminate the interference, which restricts the precoder in a null space. Thanks to the more degrees of freedom for designing $\bm W$, SDR beamforming enjoys better radar performance.
	In addition, as $\Gamma$ goes to zero, the radar performance of SDR beamforming goes to the optimal radar-only performance, while the radar performance of ZF beamforming cannot and stays constant when $\Gamma$ is lower than some value.

	\section{Numerical results} \label{section5}
	
	In this section we numerically evaluate the proposed joint beamforming methods, i.e. SDR beamforming (Algorithm~\ref{A-SDR}) and ZF beamforming (Algorithm~\ref{A1}), in a simulation study. We begin by analyzing the achievable radar beampattern of the proposed schemes, compared to the DFRC method of \cite{liu_mu-mimo_2018} in Subsection~\ref{subsec-beam-pattern}. Then, we compare ZF and SDR beamforming  in Subsection~\ref{subsec:ZF_SDR}, while the comparison of SDR beamforming and the DFRC method of \cite{liu_mu-mimo_2018} in terms of their inherent radar-communication tradeoffs is presented in Subsection~\ref{subsec:SDR_SSP}.
	
	In the experiments reported in this section, we use the following settings:
	The transmit array is a uniform linear array with half wavelength spaced elements. 
	The number of transmit antennas is $M = 10$ and the total transmit power $P_t = 1$.
	For MIMO radar transmit beamforming, the ideal beam pattern consists of three main beams, whose the directions  are $\overline{ \theta}_1 = -40 ^ \circ$, $\overline{ \theta}_2 = 0 ^ \circ$ and $\overline{ \theta}_3 = 40 ^ \circ$.
	The width of each ideal beam is $\Delta = 10 ^ \circ$, and thus the desired   beam pattern is 
	\begin{equation}
	d(\theta) = \left\{
	\begin{array}{cl}
	1, & \overline{ \theta}_p - \frac{\Delta}{2}  \leq \theta \leq   \overline{ \theta}_p + \frac{\Delta}{2},\ p = 1,2,3,  \\
	0, &  \mathrm{otherwise}.
	\end{array}
	\right.
	\label{eqn:IdealBp}
	\end{equation}
	In  \eqref{eq-2-7}, the direction grids $\left\{\theta_l \right\}_{l = 1} ^ L$ are obtained by uniformly sampling the range of $-90 ^ \circ$ to $ 90 ^ \circ $ with resolution of $0.1 ^ \circ$.
	The radar loss in  \eqref{eq-2-9} accounts for both objectives equally, namely, the weighting factor is set to $w_c = 1$.
	The multi-user communications channel obeys a Rayleigh fading model, i.e., the entries of $\bm{H}$ are i.i.d. standard complex normal random variables, and the channel output at each user is corrupted with an additive white Gaussian noise of variance  $\sigma^2 = 0.01$.
	
	In our simulations we use SINR threshold values $\Gamma$ varying from $4  \mathrm{dB}$ to $24  \mathrm{dB}$, and number of users simulated is $K = 2,4,6$.
	We simulate different $\Gamma$ and $K$ to test the impact of these parameters on the performance of the proposed joint beamforming methods.
	For each value of $\Gamma$ and $K$, the performance is averaged over $1000$ Monte Carlo tests. 
	The individual radar waveform and communication symbols comprising the transmitted signal $\bm x[n]$ in \eqref{eq-1-1} are generated as random quadrature-phase-shift-keying modulated sequences, and the transmit signal block size set to is $N = 1024$.
	
	The MATLAB CVX toolbox \cite{cvx, gb08} is used to solve the QSDP problems  \eqref{eq-5-rank} and \eqref{eq-3-34}.
	We compare our joint beamforming schemes with the DFRC beamforming method proposed in \cite{liu_mu-mimo_2018}, in which only communication symbols are precoded.
	Specially, we use gradient projection method to solve the sum-square penalty (SSP) problem under per-antenna constraint in \cite{liu_mu-mimo_2018}.
	In the sum-square penalty problem in \cite{liu_mu-mimo_2018}, the weighing factors are $\rho_1 = 1$, $\rho_2 = 2$ and the given SINR at each user is equal to the SINR threshold $\Gamma$ in \eqref{eq-3-22} .

	\subsection{MIMO Radar Transmit Beam Pattern} 
	\label{subsec-beam-pattern}
	
	First, we numerically evaluated the MIMO radar transmit beam patterns $P(\theta; \bm R)$ defined in \eqref{eq-2-3} for SDR beamforming, ZF beamforming, and the SSP approach  \cite{liu_mu-mimo_2018}.
	The transmit beam patterns for $\Gamma = 12 $ dB   are depicted in Fig.~\ref{fig:p1} for  $K = 2$ and in Fig.~\ref{fig:p2} for  $K = 4$.
	The optimal radar-only beam pattern, obtained from \eqref{eq:radaropt}, are also evaluated for comparison. 
	
	Observing Fig.~\ref{fig:p1}, we note that when $K = 2$, the average beam pattern for SDR beamforming and ZF beamforming approaches that of the optimal radar-only beamforming, while the SSP beamformer of \cite{liu_mu-mimo_2018} only  synthesizes two main beams towards $ 0 ^ \circ $ and $40 ^ \circ$.
	The fact that the SSP beamformer is unable to steer three main beams for $K = 2$ stems from its decreased MIMO radar DoF. In particular, as noted in  \eqref{eq-2-3}, the MIMO transmit beam pattern is determined by the covariance of transmit waveform, and thus the DoF for MIMO radar transmit beamforming is given by the rank of  covariance matrix.
	In the SSP approach, only communication symbols are precoded and thus the DoF cannot be larger than the number of users $K$, namely, the rank of covariance matrix cannot exceed $K$.
	In our scheme however, both communication symbols and radar waveform are precoded, and thus the DoF can be as high as its maximal value $M$, i.e., the covariance can have full rank. 
	Numerically solving \eqref{eq:radaropt} using the CVX toolbox reveals that rank of the optimal radar-only covariance $\bm{R}_0$ is $4$.
	In other words, the required DoF to achieve the optimal performance of radar is $4$.
	As a result, if $K = 2 < 4$, the SSP approach does not have enough DoF to form  three main beams as in the optimal radar beam pattern, explaining the degraded beam pattern observed in Fig.~\ref{fig:p1}.
	Our scheme are capable of forming beam patterns which are close to the optimal radar beam pattern, since the available DoF in our schemes is $M = 10 > 4$.
	When $K = 4$, the SSP approach has enough DoF and is thus capable of forming a beam pattern comparable to the optimal radar beam pattern, as shown in Fig.~\ref{fig:p2}.

	We also observe in Figs.~\ref{fig:p1}-\ref{fig:p2} that the main-lobe power of the ZF beamforming is lower than that of the SDR beamforming, implying an expected radar performance loss for ZF beamforming compared to SDR beamforming.
	This follows since when the SINR threshold $\Gamma$ is not very high, one can achieve the desired SINR level without canceling the interference, allowing to further optimize the radar beam pattern by proper optimization.
	In ZF beamforming, the interference is completely eliminated regardless of the SINR threshold, namely $\bm{H W} =  \left[ \mathrm{diag}( \sqrt{\bm{p}}),\bm{0}_{K \times M} \right]$ even if $\Gamma$ is low.
	This additional constraint limits the DoF of $\bm{W}$ and introduces the radar performance loss compared to the SDR beamforming.
	Nevertheless, in order to fully compare  ZF beamforming to SDR beamforming, one must also account for the communication performance. 
	In particular, ZF beamforming can provide improved communication rates compared SDR beamforming due to the fact that it completely eliminates the interference regardless of the specified SINR threshold.
	Our numerical results detailed in the sequel show that the obtained SINR of ZF beamforming may be much higher than $\Gamma$, while the obtained SINR by the SDR beamforming is generally quite close to $\Gamma$.
	To understand the inherent tradeoffs of the proposed schemes, in the following subsection we compare SDR and ZF beamforing in terms of both their radar and communication performance measures.
	
	
	\begin{figure} 
		\centering
		\includegraphics[width=0.9\linewidth]{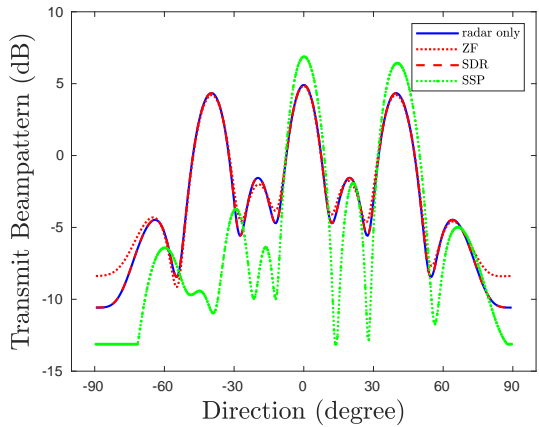}
		\caption{Transmit beam pattern of MIMO radar, for $\Gamma = 12$ dB and $K = 2$.}
		\label{fig:p1}
	\end{figure}
	
	\begin{figure} 
		\centering
		\includegraphics[width=0.9\linewidth]{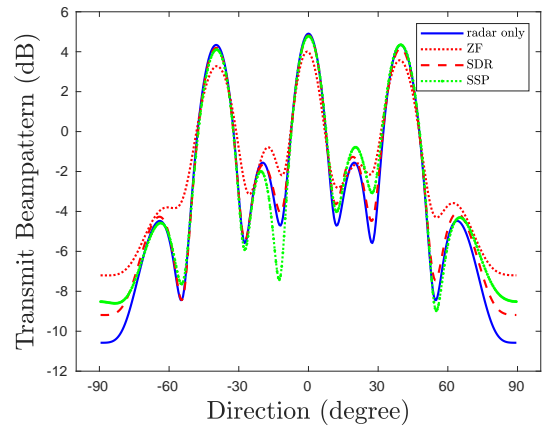}
		\caption{Transmit beam pattern of MIMO radar, for $\Gamma = 12$ dB and $ K = 4$.}
		\label{fig:p2}
	\end{figure}
	
	%

	\subsection{Comparison Between ZF and SDR Beamforming} 
	\label{subsec:ZF_SDR}
	
	In this subsection, we compare the radar performance and communication performance of SDR and ZF beamforming. 
	Radar performance is evaluated using the beam pattern MSE, defined as the MSE between the obtained MIMO radar transmit beam pattern and the optimal radar-only beam pattern, and is written as
	\begin{equation}  
	\label{eq-beam-mse}
	\mathrm{MSE} = \frac{1}{L} \sum _ {l =1} ^ L  \left| P(\theta_l; \bm{R}_0) - P(\theta_l; \bm{R}) \right| ^ 2,
	\end{equation}
	where $ P(\theta_l; \bm{R}_0) $ is the optimal radar-only beam pattern with $\bm R_0$ obtained from \eqref{eq:radaropt}. 
	Low beam pattern MSE indicates improved MIMO radar transmit beamforming.   The numerically compared beam pattern MSE values versus the SINR threshold $\Gamma$ are depicted in Fig.~\ref{fig:p4}. 
	As expected, the beam pattern MSE increases with the increment of  $\Gamma$, implying that the more restrictive SINR demands naturally come at the cost of radar performance.   
	The results in Fig.~\ref{fig:p4} validate three characters of the two proposed joint beamforming schemes: 1) SDR beamforming achieves improved radar performance compared the the sub-optimal ZF strategy; 2) The performance gap between the two methods notably narrows at high SINR constraints, i.e., as $\Gamma$ increases; 3) When $\Gamma$ is lower than some value, radar performance of the ZF beamforming stays constant, as discussed in Subsection~\ref{sub:zf:FcFr}. 
	It is also observed Fig.~\ref{fig:p4} that the more communication receivers the DFRC system has to communicate with reliably, i.e., as $K$ increases, the higher the beam pattern MSE is, again indicating the inherent tradeoff between radar and communications in DFRC systems.
	In particular, it is observed that the impact of $K$ on the beam pattern MSE is more significant than the impact of $ \Gamma $, namely, the demand to support an increased number of users is more restrictive in terms of radar performance compared to the requirement to provide improved SINR at each user.
	
	The communication performance is evaluated in terms of the achievable sum rate defined in \eqref{eq-sum-rate}. The resulting values are depicted in  Fig.~\ref{fig:sr}, where we observe that ZF beamforming achieves higher communciation rate compared to SDR beamforming, despite its performance loss for radar. This follows since, as discussed in the previous subsection, ZF beamforming typically yields SINR values higher than the imposed threshold $\Gamma$, as it nullifies the interference regardless of the value of $\Gamma$.
	Conversely, SDR beamforming, which aims at improving radar performance without imposing any structure on the resulting interference, does so by tunning its SINR to be  close to the threshold $\Gamma$, allowing to further improve radar performance without violating the SINR constraint. Hence, the achievable sum rate of SDR beamforming demonstrates an approximate linear increase with the SINR constraint $\Gamma$ in  Fig.~\ref{fig:sr}.


	From Figs.~\ref{fig:p4}  and  \ref{fig:sr}, it is observed  that performance of the two methods  coincides as $\Gamma$ increases.
	For large values of $\Gamma$, the interference tends to be naturally eliminated by the SDR beamforming in order to meet the SINR constraints. To demonstrate this property, we depict in Fig.~\ref{fig:p5}  the interference-to-noise ratio at the first user versus SINR threshold for SDR beamforming.
	Observing Fig.~\ref{fig:p5}, we note that when $\Gamma$ is high enough, the  interference power becomes much dominant  than the noise power, and thus the interference can be effectively ignored.
	Therefore, under high SINR conditions, it is reasonable to  completely eliminate the interference, and  ZF beamforming is asymptotically optimal.

	%
	
	\begin{figure}
		\centering
		\includegraphics[width=0.9\linewidth]{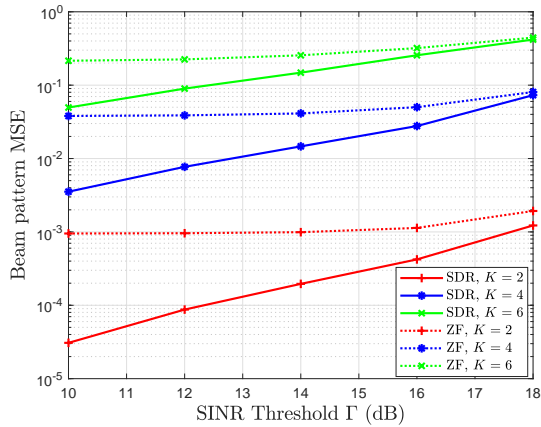}
		\caption{Beam pattern MSE versus SINR threshold $\Gamma$.}
		\label{fig:p4}
	\end{figure}
	
	\begin{figure}
		\centering
		\includegraphics[width=0.9\linewidth]{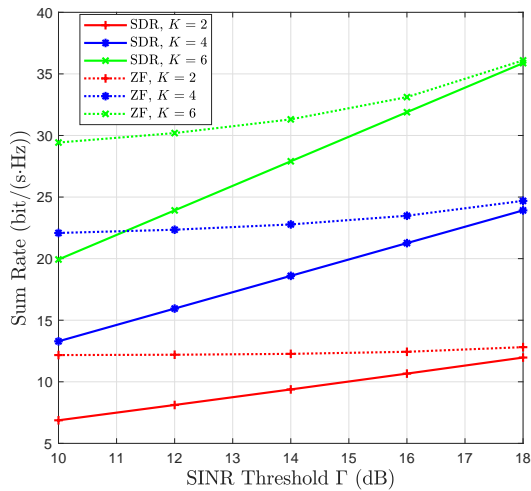}
		\caption{Achievable sum rate versus SINR threshold $\Gamma$.}
		\label{fig:sr}
	\end{figure}
	
	\begin{figure}
		\centering
		\includegraphics[width=0.9\linewidth]{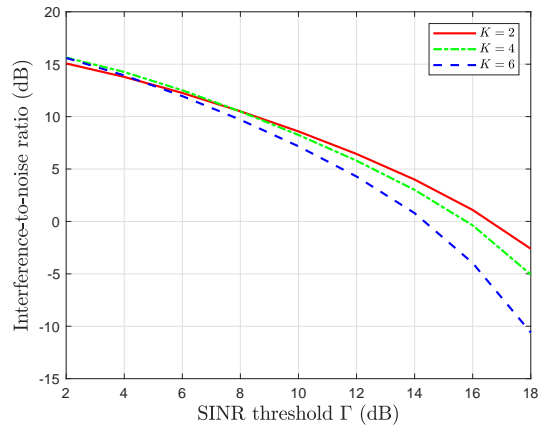}
		\caption{Radar-interference-to-noise ratio versus SINR threshold $\Gamma$, SDR beamforming.}
		\label{fig:p5}
	\end{figure}

	
	%
	We note that optimization problems tackled by SDR beamforming \eqref{eq-5-rank} and ZF beamforming \eqref{eq-3-23} are not always feasible.
	As the total transmit power is fixed to be $P_t$, the signal power at users should have an upper bound. 
	Correspondingly, the achievable feasible SINR should have an upper bound.
	If the given $\Gamma$ is too high, the joint beamforming problem \eqref{eq-5-rank}  and \eqref{eq-3-23}  may become infeasible.
	To calculate the feasible probability under a given $\Gamma$ and $K$, we ran multiple Monte Carlo tests, randomizing a new channel realization in each test. 
	The feasible  probability is calculated by dividing the number of feasible tests by the total number of tests.
	
	The relationship between the  feasible probability  and $\Gamma$  is demonstrated in Fig.~\ref{fig:p6}, for $K = 2,4,6$.
	It is observed in Fig.~\ref{fig:p6} that the feasible probability is roughly the same for ZF and SDR beamforming, and that both curves decrease as   the number of users and SINR threshold increases.
	This implies that our optimization approach may fail with very high SINR restrictions, and thus for practical applications, the SINR threshold should be carefully set.
	If the given threshold is too high, the two problems may be infeasible and our method will fail to return any meaningful solution. 
	Nevertheless, this result shows that the feasibility can almost always be ensured if $\Gamma$ is lower than some value under Rayleigh channel.
	We also note that this infeasible situation can be avoided if one changes the SINR constraints into a part of penalty functions, i.e.   \eqref{eq:sdr:objectives} or \eqref{eq:zf-iu:objectives}, as done in the scheme in \cite{liu_mu-mimo_2018}. We leave the analysis of this modification for future investigation.

	
	\begin{figure}
		\centering
		\includegraphics[width=0.9\linewidth]{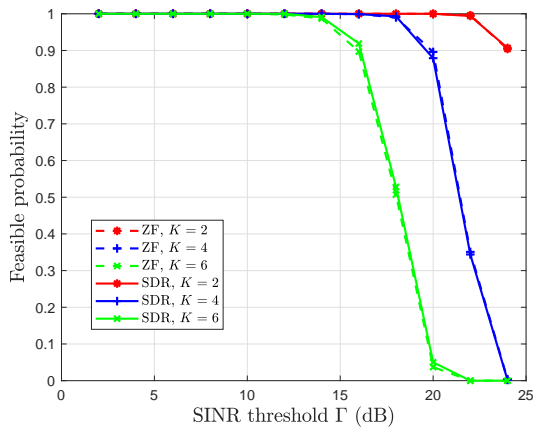}
		\caption{Feasible probability the SDR beamforming  \eqref{eq-5-rank} and ZF beamforming \eqref{eq-3-23} versus SINR threshold $\Gamma$.}
		\label{fig:p6}
	\end{figure}
	
	\subsection{Comparing SDR beamforming with SSP DFRC Method}
	\label{subsec:SDR_SSP}
	Finally, we compare our proposed SDR beamforming method to the SSP DFRC scheme  previously proposed in \cite{liu_mu-mimo_2018}. To that aim, we evaluate their tradeoff between the communication performance, encapsulated in the achieved fairness SINR defined in \eqref{eq-fainess-sinr} and the radar   beam pattern MSE  defined in \eqref{eq-beam-mse}. The numerically evaluated tradeoffs for number of users $K=2,4,6$ are depcited in Fig.~ref{fig:p8}. 
	As discussed in Subsection~\ref{subsec-beam-pattern}, our scheme notably outperform the SSP approach for $K = 2$, as clearly demonstrated in \ref{fig:p8}.  
	When $K=4,6$, our SDR beamforming technique still outperforms the SSP approach, although the gain is less notable compared to $K=2$.
	The fact that SDR beamforming outperforms the SSP method of \cite{liu_mu-mimo_2018} even when the latter is capable of exploiting the full MIMO radar DoF stems from the following reasons: 1) The SSP problem is non-convex and the obtained solution may be a local optimum; 2) In the SSP problem, the radar lost function, defined as $\| \bm{R}-\bm{R}_0 \|_F ^ 2 $, does not directly reflect the performance of radiation beam pattern.

	\begin{figure}
		\centering
		\includegraphics[width=0.9\linewidth]{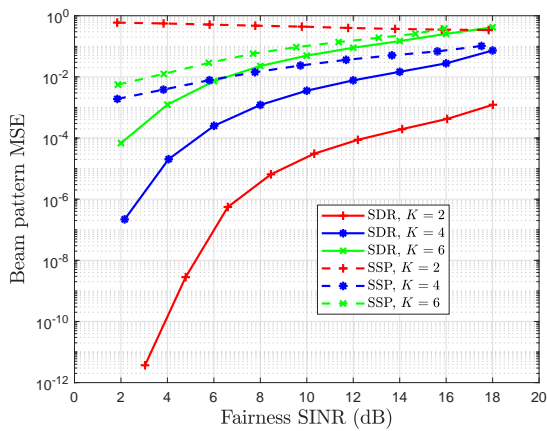}
		\caption{Beam pattern MSE versus SINR threshold $\Gamma$.}
		\label{fig:p8}
	\end{figure}
	
	Since the beam pattern MSE is not the only performance measure for radar, we also analyze the  sensing capabilities at radar receiver.
	To extract the range and angular profile of radar targets from the received radar signal, we first perform range compression \cite{li_range_2008} to obtain the range profile, and then  use the least square (LS) Capon method  \cite{xu_radar_2006,xu_target_2008} to calculate the  spatial spectrum in each range resolution bin.
	
	The first simulation is conducted to examine the range resolution and angular resolution of the MIMO radar.  
	In the simulation, there are five targets in the field of view of radar.
	The coordinate of targets in radar polar coordinate system is defined by the discrete time delay $n'$ (or the range resolution bin index) and the angular direction $\theta$ as defined in \eqref{eq-radar-receive}.
	In our parameter setting, the coordinate of these targets are $(10, 0^\circ)$, $(20, -40^\circ)$, $(20, 0^\circ)$, $(20, 40^\circ)$ and $(30, 0^\circ)$, respectively, and the  complex amplitude $\beta$ in \eqref{eq-radar-receive} for each target is $1$.
	The radar receive signal is corrupted with Gaussian noise with covariance $\bm R _ v = \sigma_r ^ 2 \bm{I} $, where $\sigma_r^2= 1$.
	The Capon spatial spectrum at the $20$-th range resolution bin and the range profile at direction $0^ \circ $ in one test are demonstrated in Fig.~\ref{fig:range_angle}, for $K = 2$ and $\Gamma = 12$dB.
	In \ref{fig:range_angle}, the range profile and Capon spatial spectrum for the radar-only case, the SSP approach and SDR beamforming are compared.
	From Fig.~\ref{fig:range_angle}, it is observed that the range and angular resolution for SDR beamforming is close to that for the radar-only case.
	The performance degradation of the SSP approach resulting from the lack of radar DoF  is significant, since the SSP approach cannot form a notable peak around the coordinate $(20, 0 ^ \circ)$, see Figs.~\ref{fig:p21} and \ref{fig:p22},    and the  amplitude estimation error at the $20$-th range resolution bin is very large, see Fig.~\ref{fig:p22}.
	When $K=2$,  the reflected signal from the three targets at the $20$-th range resolution bin in the SSP approach are linearly dependent, and thus the cross correlation defined in  \eqref{eq-2-3-1} cannot be suppressed effectively.
	Therefore, the performance of adaptive MIMO radar processing technique for the SSP approach  decreases significantly.

	\begin{figure} [htbp]
		\centering
		\subfigure[]{
			\label{fig:p11}
			\includegraphics[width=0.48\linewidth]{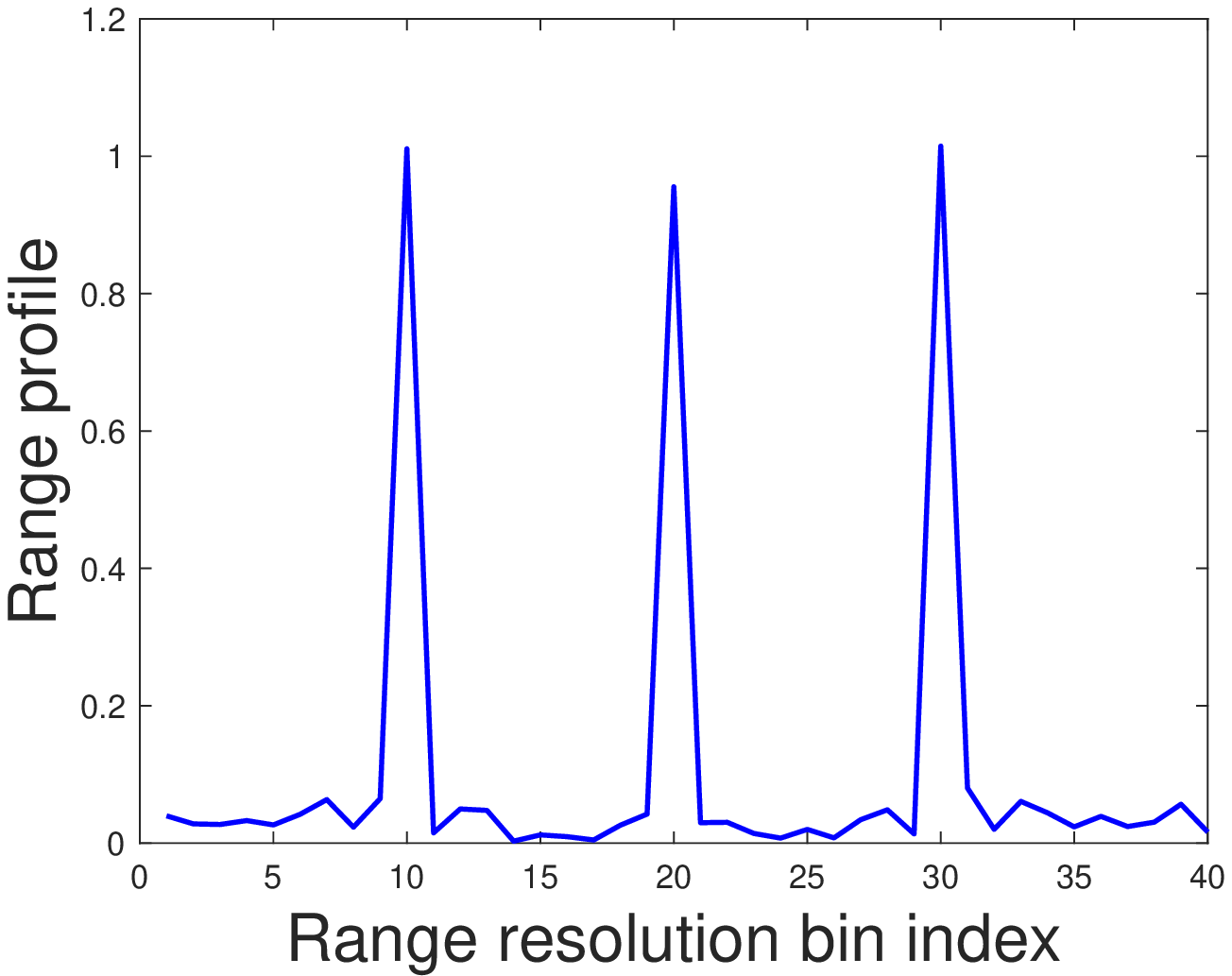}}
		\subfigure[]{
			\label{fig:p12}
			\includegraphics[width=0.48\linewidth]{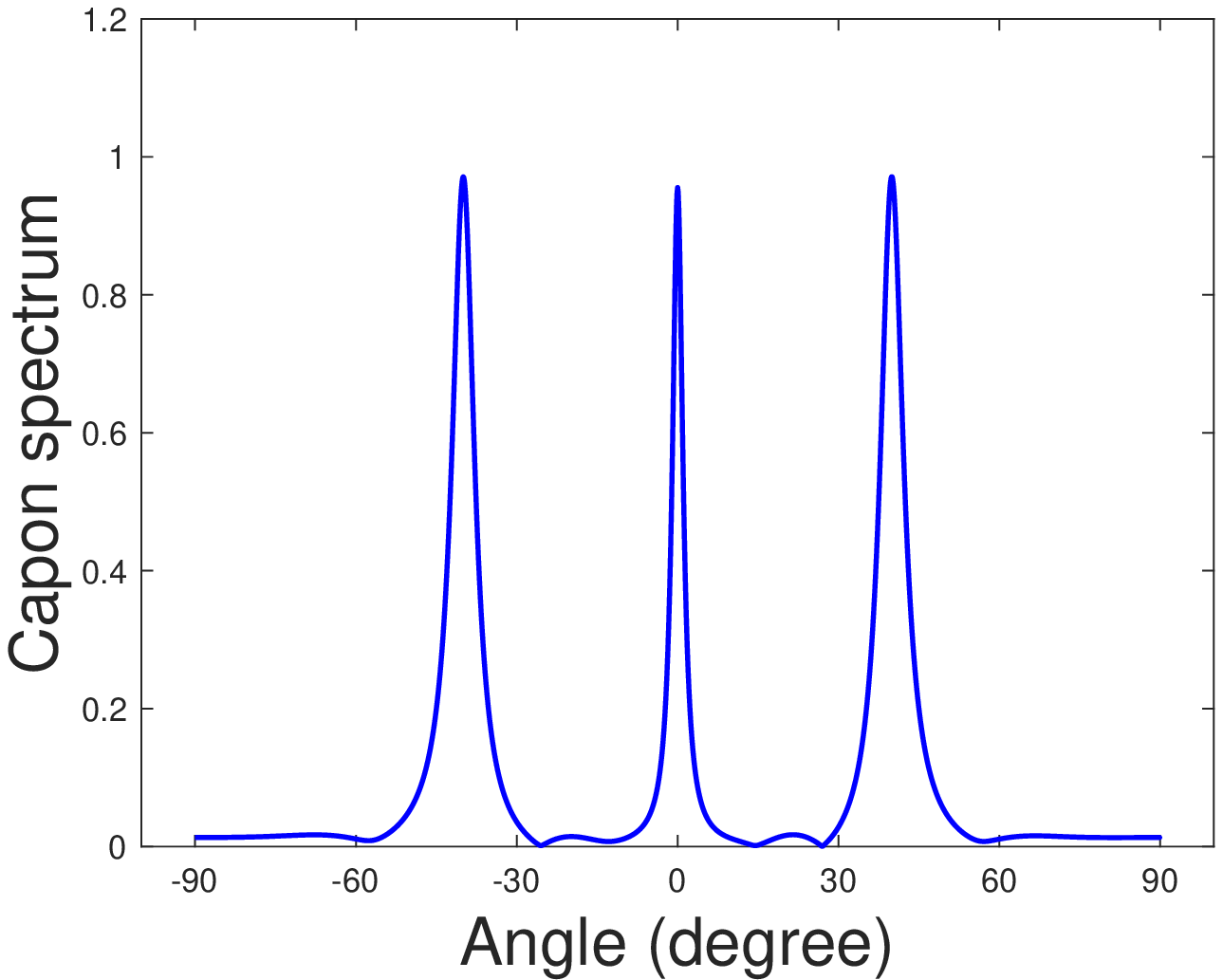}}
		\subfigure[]{
			\label{fig:p21}
			\includegraphics[width=0.48\linewidth]{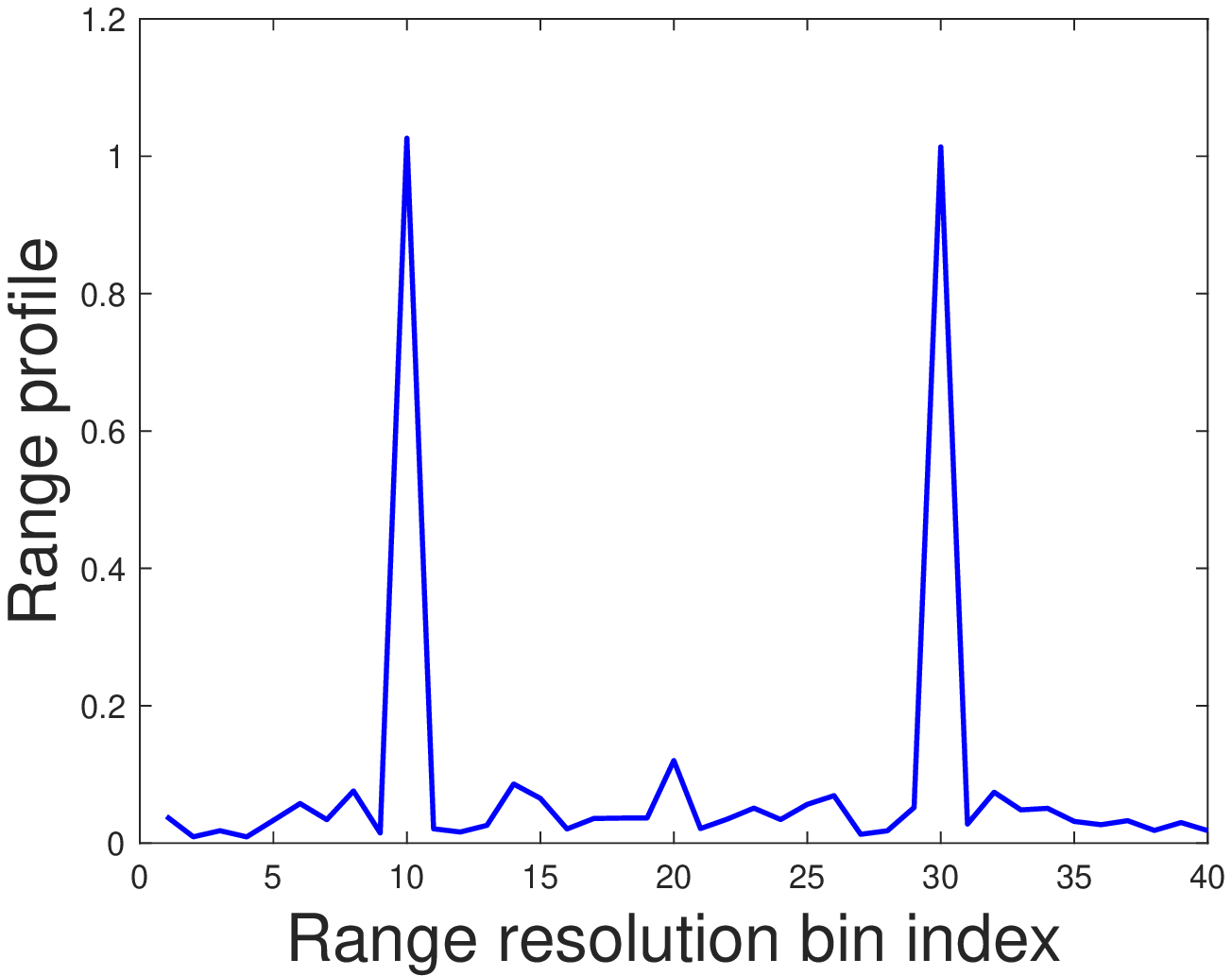}}
		\subfigure[]{
			\label{fig:p22}
			\includegraphics[width=0.48\linewidth]{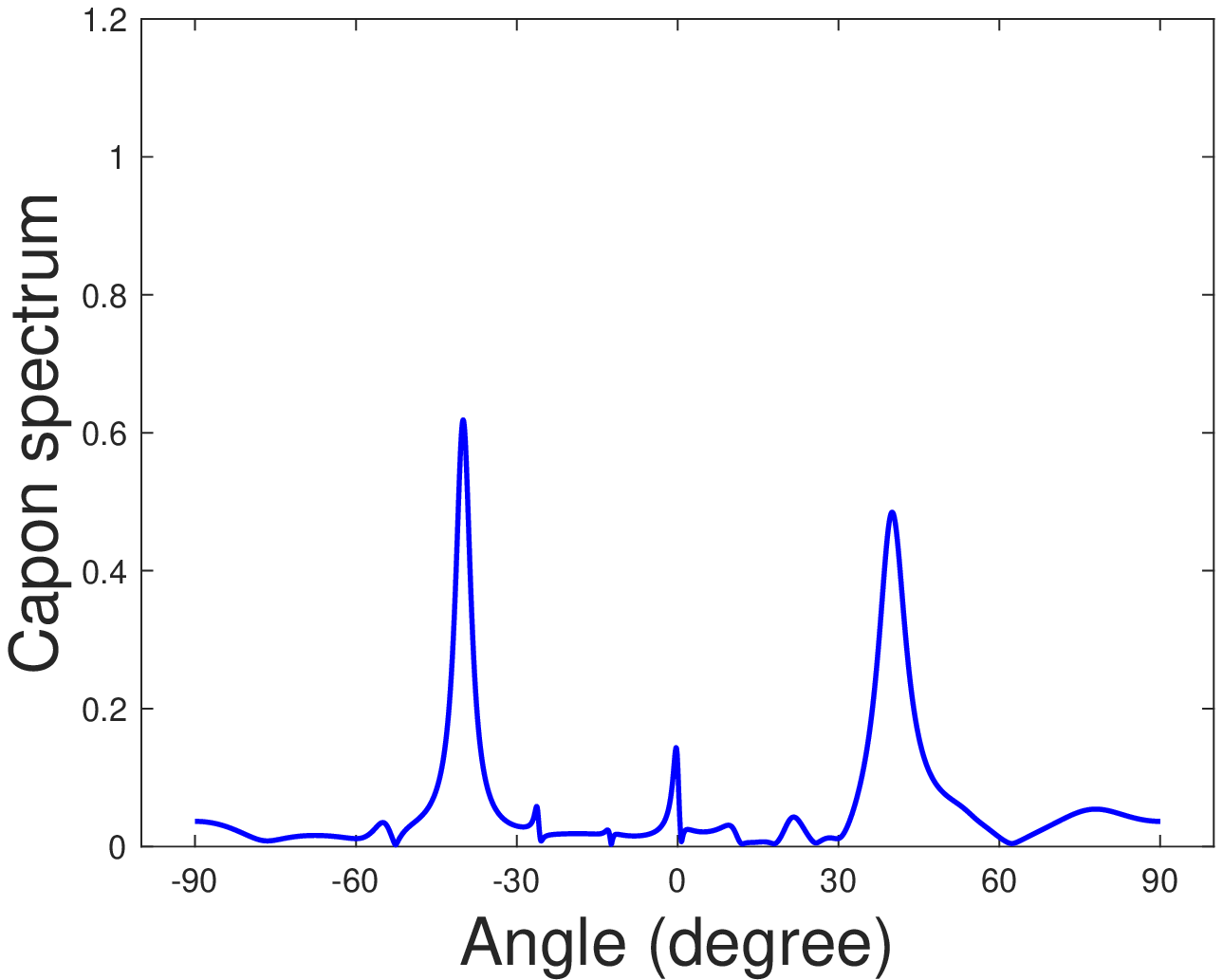}}
		\subfigure[]{
			\label{fig:p31}
			\includegraphics[width=0.48\linewidth]{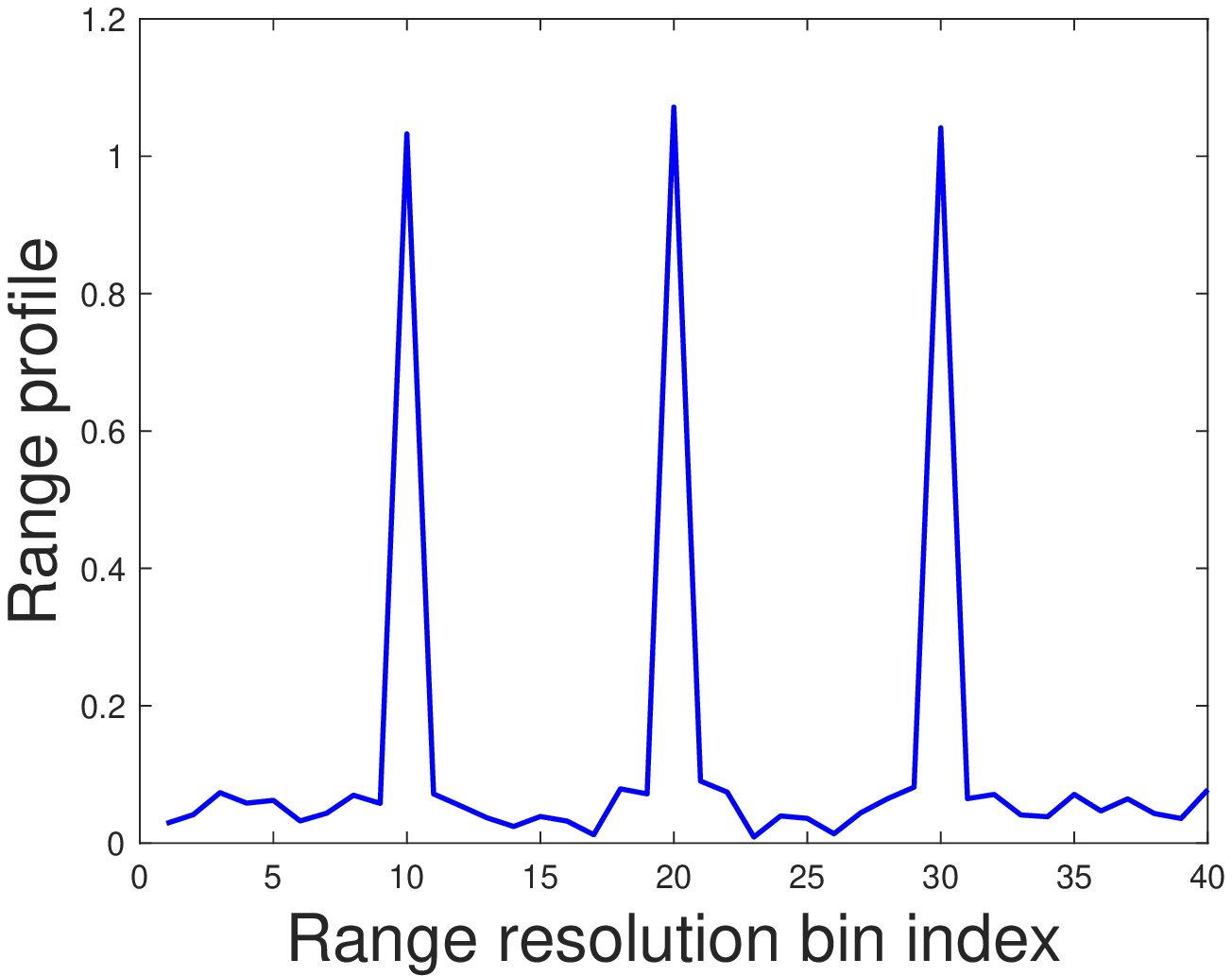}}
		\subfigure[]{
			\label{fig:p32}
			\includegraphics[width=0.48\linewidth]{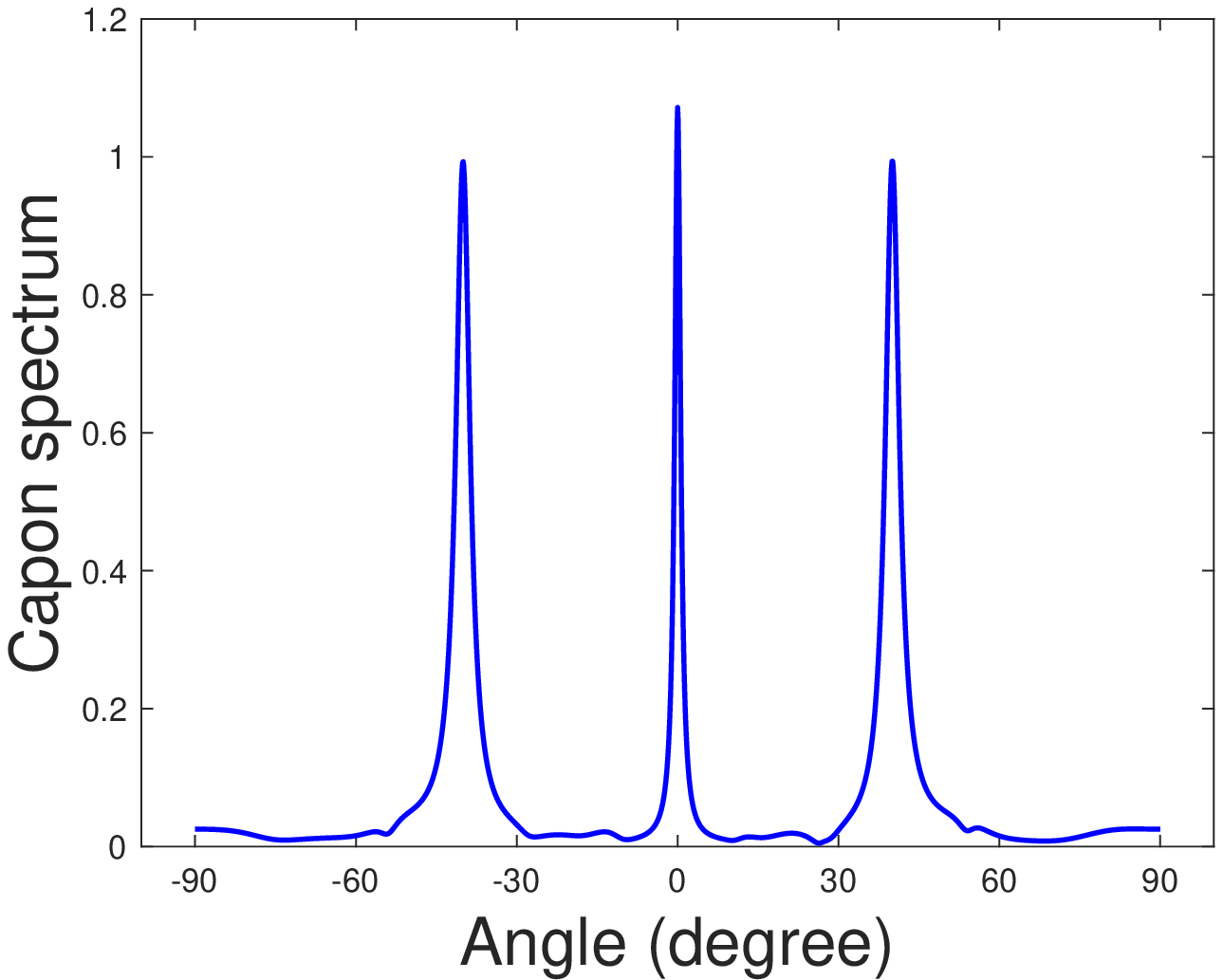}}
		\caption{Capon spatial spectrum at the $20$-th range resolution bin and the range profile at direction $0 ^ \circ$, for $K = 2$ and $\Gamma = 12$ dB. (a) Range profile for the radar-only case. (b) Capon spatial spectrum for the radar-only case. (c) Range profile for the SSP method. (d) Capon spatial spectrum for the SSP method. (e) Range profile for the SDR method. (f) Capon spatial spectrum for the SDR method.}
		\label{fig:range_angle}
		
	\end{figure}

	The second simulation is conducted to evaluate the spatial processing performance of MIMO radar, including the angle estimation accuracy and target detection performance.
	We simulate three radar targets located at directions $\overline{\theta}_1 = -40^\circ$, $\overline{\theta}_2 =  0^\circ$, and $\overline{\theta}_3 =  40^\circ$, respectively.
	These targets are in the same range resolution bin and the complex amplitude of the targets are all $1$. 
	The targets' reflected signal is corrupted with additive noise whose covariance is $\bm R _ v = \sigma_r ^ 2 \bm{I} $.
	The angle of the targets is estimated by finding the peaks of the  Capon spatial spectrum.
	The angle estimation performance  is evaluated by the root-mean-square-error (RMSE), defined as
	\begin{equation} \label{eq-simulation-rmse}
	\mathrm{RMSE} = \sqrt{ \mathbf{E} \left\{ \frac{1}{3} \sum_{p=1}^{3} ( \overline{ \theta}_p - \hat{\theta}_p ) ^ 2  \right\} } ,
	\end{equation}  
	where  $\overline{ \theta}_p $ is the real angle and $ \hat{\theta}_p  $ is the estimated angle for the $p$-th target, for $p = 1,\ldots,3$.
	The generalized  likelihood ratio test proposed in \cite{bekkerman_target_2006} is applied to detect the target.
	To demonstrate the target detection performance, we study the relationship between the detection probability and the transmit SNR given by $P_t N / \sigma_r^2$, under a fixed false alarm probability $P_{fa}$.
	To calculate the detection probability, we ran $1000$ Monte Carlo tests to produce randomized Gaussian noise for each channel realization, and thus the total number of tests is $10^6$.

	The numerically evaluated tradeoff between angle estimation RMSE and achieved fairness SINR for SDR beamforming and the SSP DFRC system of \cite{liu_mu-mimo_2018} is depicted in Fig.~\ref{fig:p7} for $K = 2,4,6$.
	Here the RMSE of the SSP approach for $K=2$ is not evaluated since it frequently fails to detect the targets near the  true angle direction of the targets.  
	The angle estimation RMSE in radar-only case is also displayed for comparison.
	Observing Fig.~\ref{fig:p7}, we note that the angle estimation RMSE tends to increase with the fairness SINR, again indicating that the improved communication performance induces some loss on the radar performance.  
	If $K = 2$, the angle estimation performance of SDR  beamforming is almost identical to the performance in radar-only case, indicating that the proposed DFRC system achieves angle estimation performance close to that of the radar-only scheme.
	The RMSE of angle estimation slightly increases if more communication users are under service. It is also noted that under most considered fairness SINR values, our proposed SDR beamforming achieves improved angle estimation RMSE compared to the SSP method. 
	
	The numerically evaluated detection probability versus transmit SNR for SDR beamforming, SSP DFRC system of \cite{liu_mu-mimo_2018} and the radar-only case is depicted in Fig.~\ref{fig:dr}, for $\Gamma = 12$ dB and $P_{fa} = 10 ^ {-4}$.
	From \cite{liu_mu-mimo_2018}, it is noted that there exists detection performance loss for simultaneous multiuser information transmission compared to the radar-only case.
	If $K=2$, the detection performance of SDR beamforming notably outperforms that of the the SSP approach, because the SSP approach usually cannot provide enough DoF to form three beams to cover the three target.
	Hence, reflected signal from one of the targets may experience notable SNR loss, significantly reducing the detection probability.
	If $K=4,6$, although the detection performance  of  SDR beamforming and the SSP approach is close, we note that the SDR beamforming can achieve better communication  quality. 
	In particular, the SDR beamforming guarantees the achieved fairness SINR is higher than $\Gamma$, while in the SSP approach the SINR is considered as a penalty term in the penalty function and the obtained SINR at users is generally less  than $\Gamma$ in our simulation.

	\begin{figure}
		\centering
		\includegraphics[width=0.9\linewidth]{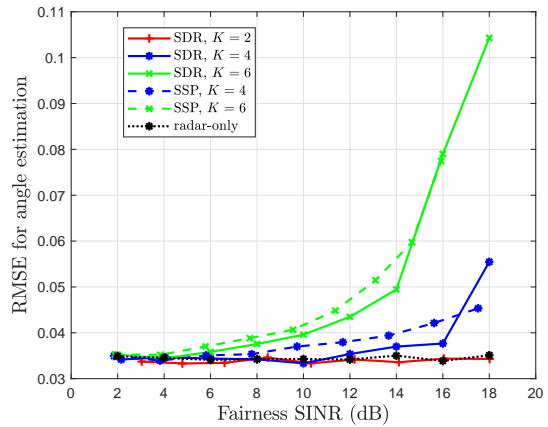}
		\caption{RMSE for angle estimation with LS-Capon method versus SINR threshold $\Gamma$.}
		\label{fig:p7}
	\end{figure}
	
	\begin{figure}
		\centering
		\includegraphics[width=0.9\linewidth]{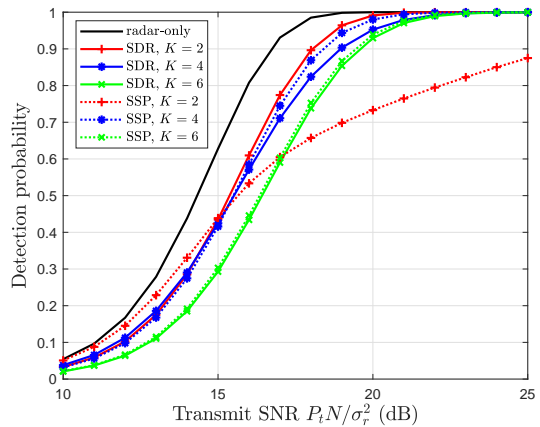}
		\caption{Detection probability versus transmit SNR under false alarm probability $P_{fa} = 10 ^ {-4}$, for $\Gamma = 12$ dB.}
		\label{fig:dr}
	\end{figure}

	\section{Conclusion}
	\label{sec:conclusion}
	In this paper, we proposed two  joint beamforming approaches for MIMO radar and multiuser MIMO communication sharing spectrum and transmit array.
	The precoders of the individual radar waveform and communication symbols are designed to optimize the performance of MIMO radar transmit beamforming while meeting SINR constraints at communication users.
	To solve the proposed optimization problem efficiently, we proposed an algorithm based on SDR, which allows to accurately recover the optimal beamforming scheme. To reduce the computational burden of SDR beamforming, we proposed a reduced complexity sub-optimal strategy based on joint beamforming with zero-forced inter-user and radar interference.
	Simulation results showed that the radar beam patterns and angle estimation performance obtained by the proposed dual-function system is comparable to those of the optimal radar-only scheme. 
	We also observed  the advantage of our method over the previous dual-function work that precodes only multiuser communication symbols from simulations in terms of their inherent radar-communication  tradeoffs. 
	These performance gains are most notable when the number of communication users is relatively small, and particularly less than the MIMO radar DoF required to meet the desired transmit beam pattern.
	

	\begin{appendix}
		\subsection{Proof of Theorem \ref{thm1}}
		\label{app:Proof1}
		Let  $ \hat{ \bm R },\hat{ \bm{R} }_{1}, \ldots, \hat{ \bm{R}}_{K}$ be an arbitrary global optimum  to  \eqref{eq-5-rank}.
		We  prove the theorem by constructing $ \tilde{ \bm R },\tilde{ \bm{R} }_{1}, \ldots, \tilde{ \bm{R}}_{K}$ from $ \hat{ \bm R },\hat{ \bm{R} }_{1}, \ldots, \hat{ \bm{R}}_{K}$ with
		\begin{equation} \label{eq-sdr-construct}
		\tilde{\bm R} = \hat{ \bm R }, \ \tilde{ \bm w }_k = \big(\bm h_k ^ H \hat{ \bm R } _ k \bm h _ k \big) ^ {-1/2} \hat { \bm R } _ k \bm h _ k, \ \tilde{ \bm R }_k = \tilde{ \bm w }_k \tilde{ \bm w }_k ^ H,
		\end{equation}
		for $k = 1,\ldots,K$.
		It is clear that $ \tilde{ \bm{R} }_{1}, \ldots, \tilde{ \bm{R}}_{K}$ are  positive semidefinite and are rank-one.
		
		We now show that $ \tilde{ \bm R },\tilde{ \bm{R} }_{1}, \ldots, \tilde{ \bm{R}}_{K}$  is also a global optimum to \eqref{eq-5-rank}.
		Since the target function $L(\bm R, \alpha)$ is determined by $\bm R$ and $ \hat{ \bm R } = \tilde{ \bm R } $, we only need to validate that $ \tilde{ \bm R },\tilde{ \bm{R} }_{1}, \ldots, \tilde{ \bm{R}}_{K}$ is a feasible solution to \eqref{eq-5-rank}.
		
		First, one can derive   that
		\begin{equation}
		\bm h _ k ^ H \tilde{ \bm R }_k \bm h _ k = \bm h _ k ^ H \tilde{ \bm w }_k \tilde{ \bm w }_k^ H \bm h _ k = \bm h _ k ^ H \hat{ \bm R }_k \bm h _ k
		\end{equation}
		by substituting \eqref{eq-sdr-construct}.
		Thus
		\begin{equation}
		\begin{aligned}
		& \left(1 + \Gamma^{-1} \right)  \bm h _ k ^ H \tilde{ \bm  R} _ k  \bm h _ k =   \left(1 + \Gamma^{-1} \right)  \bm h _ k ^ H \hat{ \bm  R} _ k  \bm h _ k \\
		& \geq \bm h _ k ^ H  \hat { \bm R }   \bm h _ k + \sigma^2  = \bm h _ k ^ H  \tilde { \bm R }   \bm h _ k + \sigma^2,
		\end{aligned}
		\end{equation}
		namely constraint \eqref{eq-5-ranke}  holds for $ \tilde{ \bm R },\tilde{ \bm{R} }_{1}, \ldots, \tilde{ \bm{R}}_{K}$.
		
		Next, we show that $ \hat{\bm R} _ k -  \tilde{ \bm R }_k \in \mathcal{S}_M^+ $. 
		For any $\bm v \in \mathbb{C}^ M$, it holds that
		\begin{equation}
		\bm  v ^ H \big( \hat{\bm R} _ k -  \tilde{ \bm R }_k \big) \bm  v = \bm  v ^ H \hat{\bm R} _ k \bm  v - \big( \bm  h_k ^ H \hat{\bm R} _ k \bm  h _ k  \big) ^ {-1} \left | \bm  v ^ H \hat{\bm R} _ k \bm  h _ k \right | ^ 2.
		\end{equation}
		According to the Cauchy-Schwarz inequality, one has
		\begin{equation}
		\big( \bm  h_k ^ H \hat{\bm R} _ k \bm  h _ k  \big)  \big( \bm  v ^ H \hat{\bm R} _ k \bm  v  \big) \geq   \left | \bm  v ^ H \hat{\bm R} _ k \bm  h _ k \right | ^ 2,
		\end{equation}
		so $  \bm  v ^ H \big( \hat{\bm R} _ k -  \tilde{ \bm R }_k \big) \bm  v \geq 0$ holds for any $\bm v \in \mathbb{C}^ M$, i.e. $ \hat{\bm R} _ k -  \tilde{ \bm R }_k \in \mathcal{S}_M^+ $.
		It therefore follows that
		\begin{displaymath}
		\tilde { \bm{R} } - \sum_{k=1}^{K} \tilde{ \bm R } _ k   = \hat { \bm{R} } - \sum_{k=1}^{K} \hat{ \bm R } _ k + \sum_{k=1}^{K} \big(  \hat{ \bm R } _ k - \tilde{ \bm R } _ k \big) \in \mathcal{S}_M^+ ,
		\end{displaymath}
		namely, the constraint \eqref{eq-5-rankb}  holds for $ \tilde{ \bm R },\tilde{ \bm{R} }_{1}, \ldots, \tilde{ \bm{R}}_{K}$.
		Finally, since  $ \hat{ \bm R } = \tilde{ \bm R }$, \eqref{eq-5-rankc} also holds for $ \tilde {\bm R} $.
		
		With the derivation above, it is verified that  $ \tilde{ \bm R },\tilde{ \bm{R} }_{1}, \ldots, \tilde{ \bm{R}}_{K}$ is  a feasible solution, and furthermore, it is also a global optimum to \eqref{eq-5-rank}, completing the proof. 
		\qed

		\subsection{Proof of Theorem \ref{thm2}}
		\label{app:Proof2}
		When the conditions \eqref{eq-3-23b} and \eqref{eq-3-23d1} hold, it follows that
		\begin{equation}
		\bm{H R H} ^ H = \bm{H W W} ^ H \bm{H} ^ H = \bm F \bm{F} ^ H ,
		\end{equation}
		i.e. \eqref{eq-3-24} holds, proving the necessity.
		
		Next, we prove that condition \eqref{eq-3-24} is also sufficient.
		Assume that condition \eqref{eq-3-24} holds. 
		We will then construct a $\bm{W}$ that satisfies  \eqref{eq-3-23b} and \eqref{eq-3-23d1}.
		To this aim, we recall that the QR decomposition \cite{zhang2017matrix} of a $n \times m$ matrix $\bm{B}$ with $n \geq m$ is defined as
		\begin{equation} \label{eq-qr-1}
		\bm{B} = \bm{P}'_a \bm{U}_a =  \bm{P} _a  \left[
		\begin{array}{c}
		\bm{U}_a \\
		\bm{0}_{(n-m) \times m}
		\end{array}   
		\right],
		\end{equation} 
		where $\bm{U}_a$ is a $m \times m$ upper triangular matrix, $\bm{P}'_a $ is a $n \times m$ matrix with orthogonal unit columns, and $\bm{P}_a$ is a $n \times n$ unitary matrix.
		Then, define the row QR decomposition of a $m \times n$ matrix $\bm{A} = \bm B^T$ with $m \le n$ as
		\begin{equation}
		\bm{A} = \bm{L}_a \bm{Q}'_a = \left[ \bm{L}_a , \bm{0}_{m \times (n - m)}\right] \bm{Q}_a,
		\end{equation}
		where $\bm{L}_a = \bm{U}_a ^ T$ is a $m \times m$ lower triangular matrix, $\bm{Q}'_a  = ( \bm{P}_a') ^ T$ is a $m \times n$ matrix with orthogonal unit rows, and $\bm{Q}_a =  \bm{P}_a ^ T$ is a $n \times n$ unitary matrix, i.e.
		\begin{equation}\label{eq:unitary}
		\bm{Q}_a \bm{Q}_a^H =\bm{Q}_a^H \bm{Q}_a = \bm I_n.
		\end{equation}
		We  note that
		\begin{equation}\label{eq:QRequ}
		\begin{aligned}
		\bm A \bm A^H &= \left[ \bm{L}_a , \bm{0}_{m \times (n - m)}\right] \bm{Q}_a \bm{Q}_a^H \left[
		\begin{array}{c}
		\bm{L}_a^H \\
		\bm{0}_{(n-m) \times m}
		\end{array}   
		\right] \\ & = \bm{L}_a \bm{L}_a^H.
		\end{aligned}
		\end{equation}

		Observing the left hand side of \eqref{eq-3-24}, we proceed by writing the Cholesky decomposition of $\bm{R}$ as
		\begin{equation} \label{eq-3-25}
		\bm{R} = \bm{L}_r \bm{L}_r ^ H,
		\end{equation}
		and writing the row QR decomposition to $\bm{H L}_r$ as
		\begin{equation} \label{eq-3-26}
		\bm{H L}_r =  \left[ \bm{L}_h ,  \bm{0}_{K \times (M-K)} \right] \bm{Q}_h,
		\end{equation} 
		where $\bm{Q}_h$ is a $M \times M$ unitary matrix and $\bm{L}_h$ is a $K \times K$ lower triangular matrix. Applying \eqref{eq:QRequ}, we rewrite the left hand side of \eqref{eq-3-24} as
		\begin{equation}\label{eq:HRHLh}
		\bm{H} \bm{R} \bm{H} ^ H = \bm{H L}_r \bm{L}_r ^ H \bm{H} ^ H = \bm{L}_h \bm{L}_h ^ H.
		\end{equation}
		
		Similarly, applying row QR decomposition to $\bm{F}$ yields
		\begin{equation} \label{eq-3-27}
		\bm{F} =  \left[ \bm{L}_f ,  \bm{0}_{K \times M} \right] \bm{Q}_f,
		\end{equation} 
		and then
		\begin{equation} \label{eq-3-28}
		\bm{F} \bm{F} ^ H = \bm{L}_f \bm{L}_f ^ H,
		\end{equation} 
		according to \eqref{eq:QRequ}. In \eqref{eq-3-27}, $\bm{Q}_f$ is a $(M+K)$-dimension unitary matrix and $\bm{L}_f$ is a $K \times K$ lower triangular matrix.
		
		Here, we note that both $\bm{H} \bm{R} \bm{H} ^ H$ and $\bm{F} \bm{F} ^ H$ are positive definite given that $\bm{F}$  is a full rank $K \times (K+M)$ matrix , indicating that the diagonal elements of $\bm{L}_h$ and $\bm{L}_f$ are all non-zero real numbers. 

		Since $\bm{L}_h$ and $\bm{L}_f$ are lower triangular matrices, we find that \eqref{eq:HRHLh} is the Cholesky decomposition of $\bm{H} \bm{R} \bm{H} ^ H$, and  \eqref{eq-3-28} is the Cholesky decomposition of $\bm{F} \bm{F} ^ H$.
		Since $\bm{H} \bm{R} \bm{H} ^ H =\bm{F} \bm{F} ^ H$ and  the Cholesky decomposition of a positive definite matrix is unique \cite{zhang2017matrix}, we have that
		\begin{equation}\label{eq:LhLf}
		\bm{L}_h = \bm{L}_f,
		\end{equation}
		if we require that the diagonal elements of $\bm{L}_h$ and $\bm{L}_f$ are positive real numbers.
		
		We can now construct the matrix $\bm{W}$ as $	\bm{W} = \bm{L}_r \bm{Q}_w $
		to satisfy \eqref{eq-3-23b}, where $\bm{Q}_w$  is a $M \times (M+K)$ matrix obeying $\bm{Q}_w \bm{Q}_w ^ H = \bm{I}_M$.
		Since we also require $\bm{W}$ to meet \eqref{eq-3-23d1}, the matrix $\bm{Q}_w$ should satisfy that
		\begin{equation}
		\bm{H} \bm{L}_r \bm{Q}_w = \bm{F}.
		\end{equation}
		To this aim, the matrix $\bm{Q}_w$ is constructed as $\bm{Q}_w = \bm{Q}_h ^ H \hat{\bm{Q}}_f$, 
		where the $M \times (M+K) $ matrix $ \hat{\bm{Q}}_f =  \left[\bm{Q}_f^T\right]^T_{{1:M}}$ denotes the first $M$ rows of $\bm{Q}_f$, and satisfies
		\begin{equation}\label{eq:Qfunitary}
		\hat{\bm{Q}}_f \hat{\bm{Q}}_f^H =  \bm I_M
		\end{equation}
		according to \eqref{eq:unitary}. 
		Thus, the matrix $\bm W$ is computed as
		\begin{equation} \label{eq-3-30}
		\bm{W} = \bm{L}_r \bm{Q}_h ^ H \hat{\bm{Q}}_f.
		\end{equation}
		Using \eqref{eq-3-30}, we can calculate $\bm W$ from $\bm R$ and $\bm F$ with $\bm L_r$, $\bm Q_h$ and $\bm Q_f$ obtained by applying  matrix decomposition, i.e. \eqref{eq-3-25}, \eqref{eq-3-26} and \eqref{eq-3-27}, respectively.
		
		To prove \eqref{eq-3-23b} and \eqref{eq-3-23d}, we substitute \eqref{eq-3-30} into these two equations, yielding
		\begin{equation}
		\bm{W} \bm{W} ^ H = \bm{L}_r \bm{Q}_h ^ H \hat{\bm{Q}}_f \hat{\bm{Q}}_f ^ H \bm{Q}_h \bm{L}_r ^ H 
		\stackrel{(a)}{=} \bm{L}_r \bm{L}_r ^ H 
		\stackrel{(b)}{=} \bm{R},
		\end{equation}
		and
		\begin{equation} \label{ap54}
		\begin{aligned}
	&	\bm{H} \bm{W} = \bm{H} \bm{L}_r \bm{Q}_h ^ H \hat{\bm{Q}}_f 
		\stackrel{(c)}{=}   \left[ \bm{L}_h ,  \bm{0}_{K \times (M-K)} \right] \bm{Q}_h \bm{Q}_h ^ H \hat{\bm{Q}}_f \\
		& \stackrel{(d)}{=}  \left[ \bm{L}_h ,  \bm{0}_{K \times (M-K)} \right]  \hat{\bm{Q}}_f
		\stackrel{(e)}{=} \left[ \bm{L}_f ,  \bm{0}_{K \times M} \right] \bm{Q}_f \stackrel{(f)}{=} \bm{F}, 
		\end{aligned}
		\end{equation}    
		respectively, where $(a)$ follows from \eqref{eq:Qfunitary} and $\bm Q_h^H \bm Q_h = \bm I_M$ \eqref{eq:unitary}; $(b)$ stems from \eqref{eq-3-25}; $(c)$ is due to \eqref{eq-3-26}; $(d)$ applies again $\bm Q_h^H \bm Q_h = \bm I_M$; $(e)$ uses \eqref{eq:LhLf}; and $(f)$ follows from \eqref{eq-3-27}. 
		Therefore, the condition \eqref{eq-3-24} is also sufficient, completing the proof.
		\qed
	\end{appendix}

	\bibliographystyle{IEEEtran}
	\bibliography{ref}
	
\end{document}